\documentclass[12pt]{amsart} 
\usepackage[english]{babel}
\usepackage[utf8]{inputenc}
\usepackage{amsmath}
\usepackage{graphicx}
\usepackage{subfigure}
\usepackage{amssymb}
\usepackage{amsthm}
\usepackage{bm}
\usepackage{tikz-cd}
\usepackage{mathrsfs}
\usepackage[colorinlistoftodos]{todonotes}
\usepackage{enumitem}
\usepackage{yfonts}
\usepackage{ dsfont }
\usepackage{bbm}
\usepackage{geometry}
\usepackage{setspace}
\usepackage{tikz}
\usepackage{pgfplots}
\usepackage{cite}
\pgfplotsset{compat=newest}
\usepgfplotslibrary{fillbetween}
\usetikzlibrary{positioning}
\usetikzlibrary{decorations.pathreplacing}
\usetikzlibrary{decorations,arrows}
\usetikzlibrary{decorations.markings}
\usetikzlibrary{patterns}
\usetikzlibrary{quotes,angles}
\usetikzlibrary{arrows}

\numberwithin{equation}{section}
\newtheorem{theorem}{Theorem}[section]
\newtheorem{lemma}[theorem]{Lemma}

\newtheorem{definition}[theorem]{Definition}

\newtheorem{remark}[theorem]{Remark}
\newtheorem{proposition}[theorem]{Proposition}

\allowdisplaybreaks[4]

\title[Nonlinear resonances in resonator crystals]{Analysis of nonlinear resonances in resonator crystals: Tight-binding approximation and existence of subwavelength soliton-like solutions}

\date{}

\author{Habib Ammari} %
\address[H. Ammari]{Department of Mathematics, ETH Z\"{u}rich, R\"{a}mistrasse 101, CH-8092 Z\"{u}rich, Switzerland; Hong Kong Institute for Advanced Study, City University of Hong Kong, Kowloon Tong, Hong Kong}
\email{habib.ammari@math.ethz.ch}

\author{Jiayu Qiu} %
\address[J. Qiu]{Department of Mathematics, ETH Z\"{u}rich, R\"{a}mistrasse 101, CH-8092 Z\"{u}rich, Switzerland}
\email{jiayu.qiu@sam.math.ethz.ch}

\begin{document}

\begin{abstract}
This work provides a mathematical framework for elucidating physical mechanisms for confining waves at subwavelength scales in periodic systems of nonlinear resonators. A discrete approximation in terms of the linear capacitance operator is provided to characterize the nonlinear subwavelength resonances. Moreover, the existence of subwavelength soliton-like localized waves in periodic systems of nonlinear resonators is proven. As a by-product,  a tight-binding approximation of the capacitance operator is shown to be valid for crystals of subwavelength resonators. Both full- and half-space crystals are considered. The framework developed in this work opens the door to the study of topological properties of periodic lattices of subwavelength nonlinear resonators, such as the emergence of nonlinearity-induced topological edge states, and to elucidate the interplay between nonlinearity and disorder.
\end{abstract}

\maketitle

\bigskip

 \par

\bigskip

\\

\tableofcontents

\section{Introduction}

Subwavelength physics is concerned with wave interactions in materials with structures at subwavelength scales. Its ultimate goal is to manipulate waves at extremely low frequencies. Recent breakthroughs, such as the emergence of the field of metamaterials (\textit{i.e.}, microstructured materials with unusual localization and transport properties), have allowed us to do this in a way that is robust and overcomes traditional diffraction limits using high-contrast resonators; see, \textit{e.g.}, \cite{paa,sheng,fink}.

In the linear case, periodic, truncated periodic, and finite subwavelength resonator systems have been intensively studied. Using first-principle mathematical analysis, new fundamental insights into the mechanisms responsible for the peculiar features of wave localization, scattering, and guiding at deep subwavelength scales have been revealed; see, \textit{e.g.}, \cite{paa,cbms}. 
The obtained results are based on the \textit{capacitance formulation}, which is a powerful tool for characterizing the subwavelength resonant modes of a system of high-contrast resonators.  The capacitance formulation provides a discrete approximation to the low-frequency part of the spectrum of the continuous PDE model, valid in the high-contrast asymptotic limit. 

Recently, there has been a shift in the approach to subwavelength physics. A class of nonlinear artificial materials consisting of nonlinear subwavelength high-contrast resonator building blocks has been introduced. This class of materials has revived the exciting prospect of realizing the long-standing goal of subwavelength robust confinement and guiding of waves (see, \textit{e.g.}, \cite{rev-acoustics,active}). 

Although very significant advances in experimental and numerical modelling have been achieved in the field of nonlinear subwavelength physics during recent decades, little is known from a mathematical point of view. In \cite{bowen-dielectric}, a rigorous mathematical framework for the analysis of nonlinear dielectric resonances in wave scattering by high-index resonators with Kerr-type nonlinearities has been provided and the existence of nonlinear dielectric resonances in the subwavelength regime, bifurcating from the zero solution at the corresponding linear resonances, has been proved. More closely to the present work, in \cite{ammari2025nonlinear_resonance}, the resonance problem for the cubic nonlinear Helmholtz equation in the subwavelength regime in the presence of a finite number of high-contrast  resonators has been considered and a discrete model has been derived to approximate the subwavelength resonances.  However, as far as we know, spectral properties of periodic nonlinear systems of subwavelength resonators have not yet been analyzed in the mathematical literature. In contrast to the linear response regime studied in depth in \cite{paa,cbms}, there is a clear lack of a deep understanding of the theory of nonlinear subwavelength resonances in periodic systems. It is important to rigorously evaluate the effect of nonlinearities on the confinement properties of periodic systems of subwavelength resonators and prove whether the properties demonstrated in the linear case persist to nonlinearity and if a new phenomenon pertaining to nonlinearity emerges.

The objective of this paper is to provide a mathematical framework for elucidating physical mechanisms for confining waves at subwavelength scales in periodic systems of nonlinear resonators. Our main focus is twofold: (i) to provide, under the assumption of high contrast in the material parameters of the resonators, a discrete approximation in terms of the linear capacitance operator to characterize the nonlinear resonances and (ii) to prove the existence of soliton-like localized waves in periodic systems of nonlinear high-contrast resonators. Such solitons are supported in the subwavelength regime. As a by-product, by establishing new properties of the \textit{Dirichlet-to-Neumann (DtN) map} corresponding to the structure that is exterior to the resonators, we show for the first time that a \textit{tight-binding approximation} of the capacitance operator holds for crystals of sbwavelength resonators. We consider both full- and half-space crystal settings. 

The paper is organized as follows. In Section \ref{sec:2}, we introduce the nonlinear resonance problem and some notation. We also state our main results on the capacitance operator, the approximation of nonlinear resonances, and the existence of subwavelength soliton-like solutions. Finally, we discuss the relation of our work to previous work on nonlinear Schr\"odinger equations. Section \ref{sec:3} is devoted to prove key properties of the capacitance operator and to derive a useful representation of the exterior DtN map. In Section \ref{sec:4}, the discrete approximation of nonlinear resonances is rigorously proved. In Section \ref{sec_gap_soliton}, the existence and characterization of soliton-like solutions are proved. Finally, some concluding remarks and open problems are given in Section \ref{sec:conclusion}. We emphasize that our results in this paper have applications in nanophotonics and nanophononics, enhanced wave-matter interactions at subwavelength scales, and active metamaterials; see, \textit{e.g.}, \cite{soliton1,soliton2,soliton3,active,nonreciprocal,rev-acoustics}.

\section{Problem setting and main results} \label{sec:2}

\subsection{Notation} \label{sec:notation}
\begin{itemize}
    \item Throughout this paper, $L^{p}(O)=L^{p}(O;\mathbb{C})$ ($p\geq 1$) denotes the complex-valued $p-$integrable functions on the open set $O$. The corresponding Sobolev space $W^{k,p}(O)$ is defined in the standard way. When $p=2$, we denote $H^{k}(O):=W^{k,2}(O)$.
    \item The complex $p-$integrable $d-$dimensional vector-valued functional space on a discrete set $S$ is denoted by $\ell^{p}(S;\mathbb{C}^{d})$. When $p=2$, $\ell^{p}(S;\mathbb{C}^{d})$ is a complex (real, resp.) Hilbert space, with its inner product defined in the standard way;
    \item $(\cdot,\cdot)_{\mathcal{H}}$ denotes the inner product of the Hilbert space $\mathcal{H}$. The $L^2(\mathcal{D})$ inner product, which is used most frequently in this paper, is written simply as $(\cdot,\cdot)_{\mathcal{D}}$. We will also abbreviate the notation for vector-valued functions as $(\nabla u,\nabla v)_{\mathcal{D}}:=\int_{\mathcal{D}}\nabla u\cdot \overline{\nabla v}$. The dual space of $H^1(\mathcal{D})$ induced by $(\cdot,\cdot)_{\mathcal{D}}$ is denoted by $(H^1(\mathcal{D}))^*$, and its corresponding norm is denoted by $\|\cdot\|_{-1,\mathcal{D}}$.
    \item The curly bracket $\langle\cdot,\cdot\rangle$ denotes the dual pairing between $H^{-1/2}(\mathcal{D})$ and $H^{1/2}(\mathcal{D})$;
    \item $\|\cdot \|_{\mathcal{B}}$ denotes the norm of the Banach space $\mathcal{B}$. In particular, the $H^k(\mathcal{D})$ ($k=0,1$) norm is denoted in the shorthand form as $\|\cdot\|_{k,\mathcal{D}}$;
    \item $\Re(z)$ (resp. $\Im(z)$) denotes the real (resp. imaginary) part of a complex number $z$.
\end{itemize}

\subsection{Formulation of the nonlinear resonance problem}
We consider the following nonlinear wave equation in $\mathbb{R}^2$:
\begin{equation} \label{eq_pde_wave_eq}
\left\{
\begin{aligned}
&\frac{\partial^2 u}{\partial t^2}-\frac{1}{n_{e}^2}\Delta u=0 \quad \text{in } \Omega:= \mathbb{R}^2\backslash \overline{\mathcal{D}}, \\
&\frac{\partial^2 u}{\partial t^2}-\frac{1}{n_{i}^2(1+\sigma |u|^2)}\Delta u=0 \quad \text{in } \mathcal{D}, \\
&u\big|_{-}=u\big|_{+},\quad \text{on } \partial \mathcal{D}, \\
&\delta \frac{\partial u}{\partial \nu}\big|_{-}=\frac{\partial u}{\partial \nu}\big|_{+} \quad \text{on } \partial \mathcal{D}.
\end{aligned}
\right.
\end{equation}
Here, the interior structure $\mathcal{D}=\cup_{\bm{n}\in\mathbb{Z}^2}(\mathcal{D}_{prim}+n_1\bm{e}_1+n_2\bm{e}_2)$ denotes the identical inclusions arranged periodically with respect to some two-dimensional lattice $\mathbb{Z}\bm{e}_1\oplus \mathbb{Z}\bm{e}_2$. The primitive inclusions $\mathcal{D}_{prim}$ are contained in the unit cell $Y:=\{t_1\bm{e}_1+t_2\bm{e}_2:0<t_1<1,0<t_2<1 \}$, and composed of several disadjoint components $\mathcal{D}_{prim}=\mathcal{D}_{prim}^{(1)}\cup \mathcal{D}_{prim}^{(2)}\cup\cdots\cup \mathcal{D}_{prim}^{(d)}$. Each $\mathcal{D}_{prim}^{(i)}$ is assumed to be simply connected and with a $C^1$ boundary. The translation of each primitive resonator is denoted by $\mathcal{D}^{(1)}_{\bm{n}}:=\mathcal{D}_{prim}^{(1)}+n_1\bm{e}_1+n_2\bm{e}_2$.

The constants $n_e^2$ (resp. $n_i^2$) in \eqref{eq_pde_wave_eq} denote the exterior (resp. interior) refractive index. $\sigma$, which introduces nonlinearity to problem \eqref{eq_pde_wave_eq}, is the Kerr constant that characterizes the dependence of the interior refractive index on the amplitude. The constant $\delta$ is positive and small. It describes the high contrast between the material parameters of the interior and exterior of the structure.

The time-harmonic solution to \eqref{eq_pde_wave_eq} is determined by the following nonlinear eigenvalue problem:
\begin{equation} \label{eq_pde_station_eq}
\left\{
\begin{aligned}
&-\frac{1}{n_{e}^2}\Delta u-\lambda u=0 \quad \text{in } \Omega, \\
&-\frac{1}{n_{i}^2}\Delta u-\lambda (1+\sigma |u|^2)u=0 \quad \text{in } \mathcal{D}, \\
&u\big|_{-}=u\big|_{+} \quad \text{on } \partial \mathcal{D}, \\
&\delta \frac{\partial u}{\partial \nu}\big|_{-}=\frac{\partial u}{\partial \nu}\big|_{+} \quad \text{on } \partial \mathcal{D}.
\end{aligned}
\right.
\end{equation}
We say that $u$ is \textit{a continuous soliton with eigenvalue $\lambda$} if $(u,\lambda)$ solves \eqref{eq_pde_station_eq} with $u\in L^2(\mathbb{R}^2)$. We are particularly interested in the soliton in the subwavelength regime where $\lambda=\mathcal{O}(\delta)$; hence throughout this paper, we will assume that $\lambda,\delta\ll 1$. For simplicity of presentation, we also set $n_{i}=1$ and $\text{vol}(\mathcal{D}_{prim}^{(i)})=1$ in the sequel.

\subsection{Main results}
The first part of this paper aims to develop a discrete approximation to solve the nonlinear eigenvalue problem \eqref{eq_pde_station_eq}. To this end, we first introduce the Dirichlet-to-Neumann (DtN) map associated with the exterior structure.
\begin{definition}[Exterior Dirichlet-to-Neumann map]
For $\phi\in H^{1/2}(\partial \mathcal{D})$, the exterior Dirichlet-to-Neumann map $\mathcal{T}^{\lambda}$ is defined as
\begin{equation*}
\begin{aligned}
\mathcal{T}^{\lambda}:\quad 
H^{1/2}(\partial \mathcal{D}) &\to H^{-1/2}(\partial \mathcal{D}), \\
\phi & \mapsto \frac{\partial u}{\partial \nu}\big|_{\partial \mathcal{D}},
\end{aligned}
\end{equation*}
where $u$ is the unique solution to the following equation:
\begin{equation} \label{eq_d2n_def}
\left\{
\begin{aligned}
&-\frac{1}{n_{e}^2}\Delta u-\lambda u=0 \quad \text{in } \Omega, \\
&u\big|_{+}=\phi \quad \text{on } \partial \Omega .
\end{aligned}
\right.
\end{equation}
\end{definition}
The map $\mathcal{T}^{\lambda}$ is well defined in the subwavelength regime (\textit{i.e.}, for $|\lambda| \ll 1$) as justified by the following proposition proved in Section \ref{sec:3}.
\begin{proposition} \label{prop_d2n_map}
There exists $\lambda_{0}>0$, which depends only on $\mathcal{D}$ and $n_{e}$, such that $\mathcal{T}^{\lambda}$ is well defined for all $|\lambda|<\lambda_0$. 
In addition, the operator-valued function $$\lambda\mapsto \mathcal{T}^{\lambda}\in \mathcal{B}(H^{1/2}(\partial \mathcal{D}),H^{-1/2}(\partial \mathcal{D}))$$ is analytic in the complex disc $\{\lambda\in\mathbb{C}:\,|\lambda|<\lambda_0\}$ and continuous on $\overline{\{\lambda\in\mathbb{C}:\,|\lambda|<\lambda_0\}}$. 
\end{proposition}

\begin{remark}
Note that for inclusions with some special geometry, the DtN map $\mathcal{T}^{\lambda}$ can be calculated explicitly; see Remark \ref{rmk_explicit_d2n}.
\end{remark}

With the DtN map $\mathcal{T}^{\lambda}$, \eqref{eq_pde_station_eq} can be transformed into the interior structure and written in the weak form:
\begin{equation} \label{eq_pde_station_weak}
    \mathfrak{a}^{cont}_{\lambda,\delta}(u,v):=\int_{\mathcal{D}}\nabla u\cdot\overline{\nabla v}-\lambda \int_{\mathcal{D}} u\cdot \overline{v}-\lambda \sigma \int_{\mathcal{D}} |u|^2 u\cdot \overline{v} -\delta \int_{\partial \mathcal{D}}\mathcal{T}^{\lambda}[u\big|_{\partial \mathcal{D}}]\cdot \overline{v}=0,
\end{equation}
for any $v\in H^1(\mathcal{D})$.

The key ingredient for the discrete approximation to solve the eigenvalue problem \eqref{eq_pde_station_eq} is the following capacitance operator associated with our periodic structure. 
\begin{definition}[Capacitance operator] \label{def_cap_operator}
For $a=(a_{\bm{n}})\in \ell^2(\mathbb{Z}^2;\mathbb{C}^{d})$, the capacitance operator $\mathcal{C}$ is defined as
\begin{equation*}
\begin{aligned}
\mathcal{C}:\quad 
\ell^2(\mathbb{Z}^2;\mathbb{C}^{d}) &\to \ell^2(\mathbb{Z}^2;\mathbb{C}^{d}), \\
(a_{\bm{n}}) & \mapsto (\sum_{\bm{m}\in \mathbb{Z}^2}\mathcal{C}_{\bm{n},\bm{m}}a_{\bm{m}}),
\end{aligned}
\end{equation*}
where the elements $\mathcal{C}_{\bm{n},\bm{m}}\in \mathbb{M}^{d\times d}$ are defined by
\begin{equation}
\label{def:capop}
\mathcal{C}_{\bm{n},\bm{m}}^{i,j}:=-\int_{\partial \mathcal{D}}\mathcal{T}^{0}[\mathbbm{1}_{\partial \mathcal{D}^{(j)}_{\bm{m}}}]\cdot \mathbbm{1}_{\partial \mathcal{D}^{(i)}_{\bm{n}}}, \quad (1\leq i,j\leq d).
\end{equation}
\end{definition}
The operator $\mathcal{C}$ is well defined and, as proved in Section \ref{sec:3}, is periodic and tight-binding in the following sense.
\begin{proposition} \label{prop_cap_operator}
The capacitance operator $\mathcal{C}$ is periodic and real valued, that is, $$\mathcal{C}_{\bm{n}+\bm{e}_k,\bm{m}+\bm{e}_k}=\mathcal{C}_{\bm{n},\bm{m}}\in\mathbb{M}^{d\times d}(\mathbb{R}) \qquad  (k=1,2).$$ Moreover, there exists $\alpha,\beta>0$ such that
\begin{equation*}
\|\mathcal{C}_{\bm{n},\bm{m}}\|\leq \alpha e^{-\beta |\bm{n}-\bm{m}|}.
\end{equation*}
\end{proposition}

Now we are prepared to introduce the discrete approximation of \eqref{eq_pde_station_eq}. We say that $a=(a_{\bm{n}})$ is \textit{a discrete soliton with eigenvalue $\lambda$} if $(a,\lambda)\in \ell^2(\mathbb{Z}^2;\mathbb{C}^{d})\times \mathbb{R}$ solves the following equation:
\begin{equation} \label{eq_discrete_station_eq}
\mathcal{C}a-\lambda(1+\sigma|a|^2)a=0.
\end{equation}
Here, $(|a|^2 a)_{\bm{n}}^{(i)}:=|a_{\bm{n}}^{(i)}|^2 a_{\bm{n}}^{(i)}$. The weak formulation of \eqref{eq_discrete_station_eq} is given by
\begin{equation}
\mathfrak{a}^{disc}_{\lambda}(a,b):=(\mathcal{C}a ,b)_{\ell^2(\mathbb{Z}^2;\mathbb{C}^{d})}-\lambda(a ,b)_{\ell^2(\mathbb{Z}^2;\mathbb{C}^{d})}-\lambda\sigma(|a|^2 a ,b)_{\ell^2(\mathbb{Z}^2;\mathbb{C}^{d})}=0,
\end{equation}
for any $b\in \ell^2(\mathbb{Z}^2;\mathbb{C}^d)$.

The first main result of this paper justifies the discrete approximation of \eqref{eq_pde_station_eq} using \eqref{eq_discrete_station_eq}. Its proof is presented in full detail in Section \ref{sec:4}.
\begin{theorem} \label{thm_tight_binding_approx}
There exist $\delta_0,\lambda_0^{cont},\lambda_0^{disc}>0$ and thresholds of amplitudes $M_{0}^{cont},M_{0}^{disc}>0$ such that the following holds. Suppose that $\delta<\delta_0$. If $u\in L^2(\mathbb{R}^2)$ is a continuous soliton with eigenvalue $\lambda^{cont}<\lambda_0^{cont}$ and satisfies $\|u\|_{L^2(\mathbb{R}^2)}\leq M_0^{cont}$, then $$a_{u}:=(\int_{\mathcal{D}_{\bm{n}}^{(1)}}u,\cdots,\int_{\mathcal{D}_{\bm{n}}^{(d)}}u)\in \ell^2(\mathbb{Z}^2;\mathbb{C}^{d})$$ is approximately a discrete soliton with eigenvalue $\lambda^{disc}=\lambda^{cont}/\delta$ in the sense that
\begin{equation*}
\sup_{\|b\|_{\ell^2(\mathbb{Z}^2;\mathbb{C}^{d})}=1}|\mathfrak{a}^{disc}_{\lambda^{disc}}(a_{u},b)|=\mathcal{O}(\delta).
\end{equation*}
Conversely, if $a\in \ell^2(\mathbb{Z}^2;\mathbb{C}^d)$ is a discrete soliton with eigenvalue $\lambda^{disc}<\lambda^{disc}_0$ and satisfies $\|a\|_{\ell^2(\mathbb{Z}^2;\mathbb{C}^d)}\leq M_0^{disc}$, then $u_a:=u^{int}_{a}\mathbbm{1}_{\mathcal{D}}+u^{ext}_{a}\mathbbm{1}_{\mathbb{R}^2\backslash\overline{\mathcal{D}}}$ is approximately a continuous soliton with eigenvalue $\lambda^{cont}=\delta\cdot \lambda^{disc}$ in the sense that
\begin{equation*}
\sup_{\|v\|_{H^1(\mathcal{D})}=1}|\mathfrak{a}^{cont}_{\lambda^{cont},\delta}(u_a,v)|=\mathcal{O}(\delta),
\end{equation*}
where $$u^{int}_{a}:=\sum_{\substack{1\leq i\leq d \\ \bm{n}\in \mathbb{Z}^2}}a_{\bm{n}}^{(i)}\mathbbm{1}_{\mathcal{D}_{\bm{n}}^{(i)}}+\mathfrak{e}_{a},$$ and $u^{ext}_{a}$ is the unique solution to \eqref{eq_d2n_def} with $\phi=u^{int}_{a}\big|_{\partial \mathcal{D}}$. Here, the remainder $\mathfrak{e}_{a}$ is defined in \eqref{eq_tight_bind_proof_14} with the estimate $\|\mathfrak{e}_a\|_{0,\mathcal{D}}=\mathcal{O}(\delta)$.
\end{theorem}

Theorem \ref{thm_tight_binding_approx} can be extended to a half-space geometry with minor changes. (In fact, our interest in half-space geometry originates from the fantastic interplay between nonlinearity and topological phenomenon \cite{li2022topological_bulk_soliton,smirnova2019topological_gap_soliton,antinucci2018universal,mukherjee2020observation,smirnova2020nonlinear}, such as the emergence of nonlinearity-induced topological edge states located at the boundary of a half-plane \cite{leykam2016edge,sone2024nonlinearity}. This is left for future study.) The half-space version (with Dirichlet boundary condition on the edge) of \eqref{eq_pde_station_eq} is

\begin{equation} \label{eq_pde_station_half_space}
\left\{
\begin{aligned}
&-\frac{1}{n_{e}^2}\Delta u-\lambda u=0\quad \text{in } \Omega^+:=(\mathbb{R}^{+}\bm{e}_1\oplus \mathbb{R}\bm{e}_2)\backslash\overline{\bigcup_{\substack{1\leq i\leq d \\ \bm{n}\in \mathbb{N}\times \mathbb{Z}}}\mathcal{D}_{\bm{n}}^{(i)}}, \\
&-\Delta u-\lambda (1+\sigma |u|^2)u=0 \quad \text{in } \mathcal{D}^{+}:=\bigcup_{\substack{1\leq i\leq d \\ \bm{n}\in \mathbb{N}\times \mathbb{Z}}}\mathcal{D}_{\bm{n}}^{(i)}, \\
&u\big|_{-}=u\big|_{+} \quad \text{on } \partial \mathcal{D}^{+}, \\
&\delta \frac{\partial u}{\partial \nu}\big|_{-}=\frac{\partial u}{\partial \nu}\big|_{+} \quad \text{on } \partial \mathcal{D}^{+}, \\
&u=0 \quad \text{on}\quad \mathbb{R}\bm{e}_2.
\end{aligned}
\right.
\end{equation}
The associated DtN map is introduced in the following definition. 

\begin{definition}[Exterior Dirichlet-to-Neumann map for the half-space setting]
For $\phi\in H^{1/2}(\partial \mathcal{D}^{+})$, the exterior Dirichlet-to-Neumann map $\mathcal{T}^{\lambda,half}$ is defined as
\begin{equation*}
\begin{aligned}
\mathcal{T}^{\lambda,half}:\quad 
H^{1/2}(\partial \mathcal{D}^{+}) &\to H^{-1/2}(\partial \mathcal{D}^{+}), \\
\phi & \mapsto \frac{\partial u}{\partial \nu}\big|_{\partial \mathcal{D}^{+}},
\end{aligned}
\end{equation*}
where $u$ is the unique solution to the following problem:
\begin{equation} \label{eq_d2n_half_space_def}
\left\{
\begin{aligned}
&-\frac{1}{n_{e}^2}\Delta u-\lambda u=0 \quad \text{in } \Omega^{+}, \\
&u\big|_{+}=\phi \quad \text{on}\quad \partial \mathcal{D}^{+} , \\
&u=0 \quad \text{on}\quad \mathbb{R}\bm{e}_2.
\end{aligned}
\right.
\end{equation}
\end{definition}
The weak formulation of \eqref{eq_pde_station_half_space} is
\begin{equation} \label{eq_pde_station_weak_half}
    \mathfrak{a}^{cont,half}_{\lambda,\delta}(u,v):=\int_{\mathcal{D}^{+}}\nabla u\cdot\overline{\nabla v}-\lambda \int_{\mathcal{D}^{+}} u\cdot \overline{v}-\lambda \sigma \int_{\mathcal{D}^{+}} |u|^2 u\cdot \overline{v} -\delta \int_{\partial \mathcal{D}^{+}}\mathcal{T}^{\lambda,half}[u\big|_{\partial \mathcal{D}^{+}}]\cdot \overline{v}=0,
\end{equation}
for any $v\in H^1(\mathcal{D}^{+})$.

The associated capacitance operator is defined as follows.
\begin{definition}[Capacitance operator for the half-space setting] \label{def_cap_operator_half_space}
For $a=(a_{\bm{n}})\in \ell^2(\mathbb{N}\times \mathbb{Z};\mathbb{C}^{d})$, the capacitance operator $\mathcal{C}^{half}$ is defined as
\begin{equation*}
\begin{aligned}
\mathcal{C}^{half}:\quad 
\ell^2(\mathbb{N}\times \mathbb{Z};\mathbb{C}^{d}) &\to \ell^2(\mathbb{N}\times \mathbb{Z};\mathbb{C}^{d}), \\
(a_{\bm{n}}) & \mapsto (\sum_{\bm{m}\in \mathbb{N}\times \mathbb{Z}}\mathcal{C}^{half}_{\bm{n},\bm{m}}a_{\bm{m}}),
\end{aligned}
\end{equation*}
where the elements $\mathcal{C}_{\bm{n},\bm{m}}\in \mathbb{M}^{d\times d}$ are defined by
\begin{equation*}
(\mathcal{C}^{half})_{\bm{n},\bm{m}}^{i,j}:=-\int_{\partial \mathcal{D}}\mathcal{T}^{0,half}[\mathbbm{1}_{\partial \mathcal{D}_{\bm{m}}^{(j)}}]\cdot \mathbbm{1}_{\partial \mathcal{D}_{\bm{n}}^{(i)}},\quad (1\leq i,j\leq d).
\end{equation*}
\end{definition}
The half-space capacitance operator is explicitly related to its whole-space counterpart when the resonator system takes a special geometry. This is illustrated in the following proposition.  
\begin{proposition} \label{prop_cap_half_vs_whole}
Suppose that the inclusions $\mathcal{D}$ are embedded in the square lattice, \textit{i.e.}, $\bm{e}_1=(1,0)$, $\bm{e}_1=(0,1)$, and are reflection symmetric in the sense that $\mathcal{F}\mathcal{D}=\mathcal{D}$, where $\mathcal{F}:(x_1,x_2)\mapsto (-x_1,x_2)$ is the reflection operator about the $y-$axis. Then, for any $\bm{n},\bm{m}\in \mathbb{N}\times \mathbb{Z}$,
\begin{equation*}
\mathcal{C}^{half}_{\bm{n},\bm{m}}=\mathcal{C}_{\bm{n},\bm{m}}-\mathcal{C}_{\bm{n},\mathcal{F}\bm{m}}.
\end{equation*}
\end{proposition}
The proof follows directly from expression \eqref{eq_d2n_2} of the whole-space 
DtN map $\mathcal{T}^{\lambda}$ (note that the half-space DtN map admits a similar expression) and the reflection symmetry of the exterior Green function. The details are left to the reader.

We say that $a=(a_{\bm{n}})$ is \textit{a discrete half-space soliton with eigenvalue $\lambda$} if $(a,\lambda)\in \ell^2(\mathbb{N}\times\mathbb{Z};\mathbb{C}^{d})\times \mathbb{R}$ solves the following equation:
\begin{equation} \label{eq_discrete_station_half_space}
\mathcal{C}^{half}a-\lambda(1+\sigma|a|^2)a=0.
\end{equation}The weak formulation of \eqref{eq_discrete_station_half_space} is
\begin{equation}
\mathfrak{a}^{disc,half}_{\lambda}(a,b):=(\mathcal{C}^{half}a ,b)_{\ell^2(\mathbb{N}\times\mathbb{Z};\mathbb{C}^{d})}-\lambda(a ,b)_{\ell^2(\mathbb{N}\times\mathbb{Z};\mathbb{C}^{d})}-\lambda\sigma(|a|^2 a ,b)_{\ell^2(\mathbb{N}\times\mathbb{Z};\mathbb{C}^{d})}=0,
\end{equation}
for any $b\in \ell^2(\mathbb{N}\times\mathbb{Z};\mathbb{C}^{d})$. 

We directly present the half-space version of Theorem \ref{thm_tight_binding_approx}, whose proof is similar to its whole-space counterpart. The constants in the following theorem, \textit{i.e.}, $\delta_0,\lambda_0^{cont},\lambda_0^{disc}>0$, \textit{etc.}, should not be taken as the same as those in Theorem \ref{thm_tight_binding_approx}. We also neglect the detailed expression of the approximate continuous soliton for brevity.

\begin{theorem} \label{thm_tight_binding_approx_half_space}
There exist $\delta_0,\lambda_0^{cont},\lambda_0^{disc}>0$ and thresholds of amplitude $M_{0}^{cont},M_{0}^{disc}>0$ such that the following holds. Suppose that $\delta<\delta_0$. If $u\in L^2(\mathbb{R}^{+}\bm{e}_1\oplus \mathbb{R}\bm{e}_2)$ is a continuous half-space soliton with eigenvalue $\lambda^{cont}<\lambda_0^{cont}$ and satisfies $\|u\|_{L^2(\mathbb{R}^{+}\bm{e}_1\oplus \mathbb{R}\bm{e}_2)}\leq M_0^{cont}$, then $$a_{u}:=(\int_{\mathcal{D}_{\bm{n}}^{(1)}}u,\cdots,\int_{\mathcal{D}_{\bm{n}}^{(d)}}u)\in \ell^2(\mathbb{N}\times \mathbb{Z};\mathbb{C}^{d})$$ is approximately a discrete half-space soliton with eigenvalue $\lambda^{disc}=\lambda^{cont}/\delta$ in the sense that
\begin{equation*}
\sup_{\|b\|_{\ell^2(\mathbb{N}\times \mathbb{Z};\mathbb{C}^{d})}=1}|\mathfrak{a}^{disc,half}_{\lambda^{disc}}(a_{u},b)|=\mathcal{O}(\delta).
\end{equation*}
Conversely,
if $a\in \ell^2(\mathbb{N}\times \mathbb{Z};\mathbb{C}^{d})$ is a discrete half-space soliton with eigenvalue $\lambda^{disc}<\lambda_0^{disc}$ and satisfies $\|a\|_{\ell^2(\mathbb{N}\times \mathbb{Z};\mathbb{C}^{d})}\leq M_0^{disc}$, then there exists $$u_a=\sum_{\substack{1\leq i\leq d \\ \bm{n}\in \mathbb{N}\times\mathbb{Z}}}a_{\bm{n}}^{(i)}\mathbbm{1}_{\mathcal{D}_{\bm{n}}^{(i)}}+\mathcal{O}(\delta),$$ which is approximately a continuous half-space soliton with eigenvalue $\lambda^{cont}=\delta\cdot \lambda^{disc}$ in the sense that
\begin{equation*}
\sup_{\|v\|_{H^1(\mathcal{D}^{+})}=1}|\mathfrak{a}^{cont,half}_{\lambda^{cont},\delta}(u_a,v)|=\mathcal{O}(\delta).
\end{equation*}
\end{theorem}

The second part of this paper considers the existence of discrete gap solitons with a focus on nonlinearity, \textit{i.e.}, the localized solutions to \eqref{eq_discrete_station_eq} or \eqref{eq_discrete_station_half_space} with $\lambda$ lying in a spectral gap of the capacitance operator $\mathcal{C}$. The mechanism leading to the emergence of gap solitons in \textit{periodic structures with focusing nonlinearity} is well understood: the lower-gap and upper-gap spectrum of the governing operator is linked by the nonlinear term and forms a valley-top energy landscape, which leads to the existence of critical points and hence gap solitons. In fact, following the lines in \cite{pankov2010soliton_discrete,pelinovsky2011localization} on the existence of gap solitons of discrete nonlinear Schrödinger equations (DNLS), we can prove the following result. 
\begin{theorem} \label{thm_gap_solitons}
Suppose that the capacitance operator $\mathcal{C}$ ($\mathcal{C}^{half}$, resp.) has a spectral gap $\mathcal{I}$ in the sense that
\begin{equation*}
\inf \mathcal{I}>0,\quad \text{Spec}(\mathcal{C})\cap \mathcal{I}=\emptyset,\quad \text{Spec}(\mathcal{C})\cap (-\infty,\inf \mathcal{I})\neq \emptyset,\quad \text{Spec}(\mathcal{C})\cap (\sup \mathcal{I},\infty)\neq \emptyset ,
\end{equation*}
where $\text{Spec}$ denotes the spectrum. Then, for any $\lambda\in\mathcal{I}$ and $\sigma>0$, equation \eqref{eq_discrete_station_eq} (\eqref{eq_discrete_station_half_space}, resp.) has a nontrivial real-valued solution $u\in \ell^2(\mathbb{Z}^2;\mathbb{C}^{d})$ ($\ell^2(\mathbb{Z}\times \mathbb{N};\mathbb{C}^{d})$, resp.), which decays exponentially at infinity.
\end{theorem}
In this paper, we extend the study of gap solitons to \textit{periodic structures with localized defects}. To be precise, we consider the following problem:
\begin{equation} \label{eq_discrete_station_eq_defect}
\mathcal{C}a+Va-\lambda(1+\sigma|a|^2)a=0,
\end{equation}
where the defect operator $V$ is assumed to be self-adjoint on $\ell^2(\mathbb{Z}^2;\mathbb{C}^{d})$ and bounded from $\ell^{\infty}(\mathbb{Z}^2;\mathbb{C}^{d})$ to $\ell^{1}(\mathbb{Z}^2;\mathbb{C}^{d}))$. In fact, any real-valued function $V\in \ell^{1}(\mathbb{Z}^2;\mathbb{C}^{d}))$ satisfies these conditions. The main result of the second part of this paper is the following. 
\begin{theorem} \label{thm_defect_gap_solitons}
Suppose that the capacitance operator $\mathcal{C}$  has a spectral gap $\mathcal{I}$ in the same sense as in Theorem \ref{thm_gap_solitons}. Then for any $\lambda\in\mathcal{I}$ and $\sigma>0$, if we have
\begin{equation*}
2\|V\|_{\mathcal{B}(\ell^{\infty}(\mathbb{Z}^2;\mathbb{C}^{d}),\ell^{1}(\mathbb{Z}^2;\mathbb{C}^{d})))}<\text{dist}\{\lambda,\text{Spec}(\mathcal{C})\},
\end{equation*}
then equation \eqref{eq_discrete_station_eq_defect} has a nontrivial solution $u\in \ell^2(\mathbb{Z}^2;\mathbb{C}^{d})$, which decays exponentially at infinity.
\end{theorem}

The proof is given in full detail in Section \ref{sec_gap_soliton}. The key insight is that when the size of the defect $V$ is suitably controlled, it only slightly deforms the energy landscape associated with the nonlinear problem \eqref{eq_discrete_station_eq_defect} without breaking the valley-top structure. This leads to the existence of critical points of the energy functional, or equivalently, solutions to \eqref{eq_discrete_station_eq_defect}. We also note that the proof of Theorem \ref{thm_defect_gap_solitons} is directly generalized to the half-space case, if the half-space capacitance operator is spectrally gapped, as in Theorem \ref{thm_gap_solitons}. We leave this extension to the reader.

\subsection{Relation to previous works}

As one of the major contributions of this work, the discrete approximation of nonlinear subwavelength resonance in periodic structures presented in Theorem \ref{thm_tight_binding_approx} and \ref{thm_tight_binding_approx_half_space} extends previous results in \cite{feppon-dtn,feppon2022modal} and \cite{ammari2025nonlinear_resonance}, which justifies a similar discrete approximation of the subwavelength resonance problem but for linear periodic and nonlinear finite structures, respectively. Note that the use of capacitance operator (or matrix) lies in the center of the discrete approximation presented in both these mentioned works and in the present paper. We note that, in the physics literature, the discrete approximation is one of the most widely applied techniques to study either classical or quantum wave systems. Besides the aforementioned studies, we note that the discrete approximation of (linear or nonlinear, stationary or time-dependent, periodic or non-periodic) Schrödinger equations has been rigorously justified in \cite{ablowitz2012tight_bind,pelinovsky2010bounds_tight-binding,gilg2024approximation,sacchetti2020derivation_tight-binding,shapiro2022tight_binding,pelinovsky2008justification,bambusi2007exponential_time,aftalion2009mathematical}; see also the references therein. We note that, with the results in Theorem \ref{thm_tight_binding_approx}, it is expected that the discrete approximation of the time-dependent equation \eqref{eq_pde_wave_eq} will be justified. To be more precise, one can consider the Cauchy problem associated with the following equation:
\begin{equation*}
\frac{\partial^2 a}{\partial t^2}+\frac{1}{1+\sigma|a|^2}a=0,
\end{equation*}
(the fraction $\frac{1}{1+\sigma|a|^2}$ is understood pointwisely, \textit{i.e.}, $\Big(\frac{1}{1+\sigma|a|^2}\Big)_{\bm{n}}^{(i)}=\frac{1}{1+\sigma|a_{\bm{n}}^{(i)}|^2}$) and show that its solution approximates the solution of \eqref{eq_pde_wave_eq} with an error estimate holding in long (but finite) time; see, \textit{e.g.}, \cite{ablowitz2012tight_bind}. We leave this extension to the reader.

Another major contribution of this paper, Theorem \ref{thm_defect_gap_solitons}, is on the existence of subwavelength band-gap solitons in periodic discrete structures with localized defects. We note that the gap solitons has been extensively studied for the nonlinear Schrödinger equation (NLS) \cite{Pankov2005soliton_continuous,weinstein2010edge_bifurcation,dohnal2016bifurcation,pelinovsky2011localization,jeanjean2003remark,lions1984concentration,pava2009existence,stuart1997bifurcation,alama1992existence,heinz1992existence,kupper1992necessary} and its discrete counterparts (DNLS) \cite{pankov2010soliton_discrete,shi2010existence_gap_soliton,chen2014homoclinic,zhou2010existence_homoclinic,lin2022homoclinic,chen2024discrete,zhou2010existence_gap_solitons,kuang2014existence,hofstrand2023discrete}. Among the mathematical studies in discrete solitons, despite a wide range of structures being considered, including different types of nonlinearity (superlinear- and saturable-type, etc.), most works focus on solitons embedded in a periodic lattice with energy lying in (finite or half-infinite) spectral gaps. In \cite{chen2019perturbed,zhang2009breather}, the authors prove the existence of solitons in non-periodic structures perturbed by unbounded potentials. In \cite{sacchetti2023perturbation}, the author proves the bifurcation of solitons from an eigenvalue of the defected structure (periodic NLS perturbed by a bounded potential) from the perspective of perturbation theory. Our result focuses on a different perspective, that is, discrete solitons with energy lying in a \textit{finite band gap} of the background lattice that also possesses \textit{a localized defect}. We believe that our work provides a valuable contribution to the mathematical study of discrete solitons. We also remark that the proof of Theorem \ref{thm_defect_gap_solitons} can be generalized to other systems like DNLS with different types of nonlinearity; this interesting extension is left to the reader.

\section{Exterior Dirichlet-to-Neumann map and capacitance operator} \label{sec:3}

\subsection{Proof of Proposition \ref{prop_d2n_map}}

The idea of proof proceeds as follows. We first introduce the Laplacian operator $\mathcal{L}^{ext}$ in the exterior structure with Dirichlet boundary condition, \textit{i.e.},
\begin{equation*}
\mathcal{L}^{ext}: H^{1}_{0}(\Omega)\subset L^2(\Omega)\to H^{-1}(\Omega),\quad
u\mapsto -\frac{1}{n_{e}^2}\Delta u .
\end{equation*}
Then, based on the Floquet transform and the Poincaré inequality, we will show that $\mathcal{L}^{ext}$ has a positive spectrum due to the imposed Dirichlet boundary condition. Denoting $$\lambda_0:=\frac{1}{2}\inf \text{Spec}(\mathcal{L}^{ext})>0,$$ we will express $\mathcal{T}^{\lambda}$ based on the resolvent $(\mathcal{L}^{ext}-\lambda)^{-1}$ ($|\lambda|<\lambda_0$), which immediately leads to the boundedness and analyticity of $\mathcal{T}^{\lambda}$.

{\color{blue}Step 1.} By the Floquet transform, the spectrum of $\mathcal{L}^{ext}$ is decomposed as $$\text{Spec}(\mathcal{L}^{ext})=\cup_{\bm{\kappa}\in Y^{*}}\text{Spec}(\mathcal{L}^{ext}_{\bm{\kappa}}),$$ where $\mathcal{L}^{ext}_{\bm{\kappa}}$ is the Laplacian operator on the $\bm{\kappa}-$quasi-periodic functional space
\begin{equation*}
\mathcal{U}_{\bm{\kappa}}H^{1}_{0}(\Omega):=\big\{u\in H_{loc}^1(\Omega):\, u(\bm{x}+\bm{e}_k)=e^{i\bm{\kappa}}u(\bm{x}),\,(\nabla u)(\bm{x}+\bm{e}_k)=e^{i\bm{\kappa}}(\nabla u)(\bm{x}),\, u\big|_{\partial \mathcal{D}}=0\big\} .
\end{equation*}
It is well known that $\mathcal{L}^{ext}_{\bm{\kappa}}$ has an unbounded, discrete, and positive spectrum. We claim that there exists a uniform positive lower bound: 
\begin{equation} \label{eq_ext_inf_spectrum}
\inf_{\bm{\kappa}\in Y^{*}}\inf\text{Spec}(\mathcal{L}^{ext}_{\bm{\kappa}})>0.
\end{equation}
Then it follows that $$\inf \text{Spec}(\mathcal{L}^{ext})=\inf_{\bm{\kappa}\in Y^{*}}\inf\text{Spec}(\mathcal{L}^{ext}_{\bm{\kappa}})>0.$$ Denoting $\lambda_0=\frac{1}{2}\inf \text{Spec}(\mathcal{L}^{ext})$, one sees that the resolvent $(\mathcal{L}^{ext}-\lambda)^{-1}$ is well defined for $|\lambda|<\lambda_0$.

Now we prove \eqref{eq_ext_inf_spectrum}. Suppose that $u\in \mathcal{U}_{\bm{\kappa}}H^{1}_{0}(\Omega)$ and $\mathcal{L}^{ext}_{\bm{\kappa}}u=\lambda u$. We first prove a Poincaré-type estimate of functions in $\mathcal{U}_{\bm{\kappa}}H^{1}_{0}(\Omega)$. Extending $u$ by zero inside $\mathcal{D}_{prim}$, fixing an arbitrary point $O\in \mathcal{D}_{prim}$ (without loss of generality, we assume that $O\in \mathcal{D}_{prim}^{(1)}$), the Dirichlet boundary condition on $\mathcal{D}_{prim}$ implies that
\begin{equation*}
u(\bm{y})=\int_{0}^{t}\frac{\partial u}{\partial \hat{l}}(\bm{x}+s\hat{l})ds ,  \quad \forall \bm{y}\in Y\cap \Omega .
\end{equation*}
Here, $\bm{x}$ is the closest point to $\bm{y}$ on the intersection $l\cap \partial \mathcal{D}_{prim}^{(1)}$ (the segment $l=l_{O\bm{y}}$ connects $O$ and $\bm{y}$; If $\mathcal{D}_{prim}^{(1)}$ is star shaped with $O$ being its center, then $\{\bm{x}\}=l\cap \partial \mathcal{D}_{prim}^{(1)}$), $\hat{l}$ denotes the direction of $l$, and $t=|\bm{y}-\bm{x}|$. The Cauchy-Schwarz inequality yields
\begin{equation*}
|u(\bm{y})|^2\leq \text{diam}(Y)\int_{0}^{t}\Big|\frac{\partial u}{\partial \hat{l}}(\bm{x}+s\hat{l})\Big|^2 ds 
\leq \text{diam}(Y)\int_{0}^{t}\Big|\nabla u(\bm{x}+s\hat{l})\Big|^2 ds .
\end{equation*}
The integral on the right-hand side can be rewritten in polar form. Denoting the distance from any point to the origin $O$ by $r$ and the corresponding argument by $\theta$, we see that
\begin{equation*}
|u(\bm{y})|^2\leq \text{diam}(Y)\int_{r_1(\theta)}^{r_2(\theta)}|\nabla u(r,\theta)|^2dr 
\leq \frac{\text{diam}(Y)}{\text{dist}(O,\partial \mathcal{D}_{prim}^{(1)})} \int_{r_1(\theta)}^{r_2(\theta)}|\nabla u(r,\theta)|^2rdr ,
\end{equation*}
where $r_1$ (resp., $r_2$) is the distance from $O$ to the closest point on $l\cap \partial \mathcal{D}_{prim}^{(1)}$ (to $l\cap \partial Y$, respectively). Note that the integral $\int_{r_1}^{r_2}|\nabla u(r,\theta)|^2rdr$ depends only on the argument of $l=l_{O\bm{y}}$ and not on the length $|l_{O\bm{y}}|$. Hence, by integrating over $\bm{y}$ and calculating this integral in the polar coordinates, we obtain the following Poincaré-type inequality:
\begin{equation*}
\int_{Y\cap \Omega}|u(\bm{y})|^2d\bm{y} 
\leq  \frac{\text{diam}^3(Y)}{\text{dist}(O,\partial \mathcal{D}_{prim}^{(1)})} \int d\theta\int_{r_1(\theta)}^{r_2(\theta)}|\nabla u(r,\theta)|^2rdr 
=C_0 \int_{Y\cap \Omega}|\nabla u(\bm{y})|^2d\bm{y} ,
\end{equation*}
where $$C_0:=\frac{\text{diam}^3(Y)}{\text{dist}(O,\partial \mathcal{D}_{prim}^{(1)})}.$$ 

From the identity $\mathcal{L}^{ext}_{\bm{\kappa}}u=\lambda u$, integrating by parts over $Y\cap \Omega$, and using the quasi-periodicity of $u$ and the above inequality, we obtain
\begin{equation*}
\lambda \int_{Y\cap \Omega}|u(\bm{y})|^2d\bm{y}=\int_{Y\cap \Omega}|\nabla u(\bm{y})|^2d\bm{y}\geq C_{0}^{-1} \int_{Y\cap \Omega}| u(\bm{y})|^2d\bm{y} .
\end{equation*}
Note that $C_0>0$ is independent of $\bm{\kappa}$, which concludes the proof of \eqref{eq_ext_inf_spectrum}.

{\color{blue}Step 2.} With the result in Step 1 in hand, the uniqueness of a solution to \eqref{eq_d2n_def} is obvious. In fact, suppose that both $u$ and $\tilde{u}$ solve \eqref{eq_d2n_def} and are not identical. This leads to the contradiction that $w:=u-\tilde{u}\in L^2(\Omega)$ is an eigenfunction of $\mathcal{L}^{ext}$ with an eigenvalue $\lambda$ in $\mathcal{R}\backslash \text{Spec}(\mathcal{L}^{ext})$.

Now we prove the existence of a solution to \eqref{eq_d2n_def}. Denote the trace operator from $H^1(\Omega)$ to $ H^{1/2}(\partial \Omega)$ by $T$. We can construct a right inverse $E$ of $T$ (lifting operator) as follows:
\begin{enumerate}
    \item[(i)] take an open neighborhood $\tilde{\mathcal{D}}^{(i)}_{prim}$ of $\mathcal{D}^{(i)}_{prim}$ for each $i$ such that $\mathcal{D}^{(i)}_{prim}\Subset\tilde{\mathcal{D}}^{(i)}_{prim}\Subset Y$ and $\mathcal{D}^{(i)}_{prim}\cap \mathcal{D}^{(j)}_{prim}=\emptyset$ if $i\neq j$;
    \item[(ii)] construct a lifting map $E^{(i)}$ from $ H^{1/2}(\partial \mathcal{\mathcal{D}}_{prim}^{(i)})$ into $H_{0}^{1}(\tilde{\mathcal{D}}^{(i)}_{prim}\backslash \mathcal{D}^{(i)}_{prim})$ following a patching argument, \textit{e.g.}, as in \cite[Theorem 18.18]{leoni2024first_course_sobolev};
    \item[(iii)] repeat the same procedure in each cell $Y_{\bm{n}}=Y+n_1\bm{e}_1+n_2\bm{e}_2$, and denote the resulting lifting map by $E^{(i)}_{\bm{n}}$;
    \item[(iv)] add the resulting maps and obtain $E:=\sum_{\substack{\bm{n}\in\mathbb{Z}^2 \\ 1\leq i\leq d}}E^{(i)}_{\bm{n}}\circ \mathbbm{1}_{\partial \mathcal{D}_{\bm{n}}^{(i)}}$.
\end{enumerate}
Next, we construct the solution to \eqref{eq_d2n_def} using a standard lifting argument. We let
\begin{equation*}
f_{\lambda,\phi}:=(\frac{1}{n_{e}^2}\Delta +\lambda)(E\phi)\in H^{-1}(\Omega).
\end{equation*}
Then we set $g_{\lambda,\phi}:=(\mathcal{L}^{ext}-\lambda)^{-1}f_{\lambda,\phi}$, which uniquely solves the following problem: 
\begin{equation*}
\left\{
\begin{aligned}
&-\frac{1}{n_{e}^2}\Delta g_{\lambda,\phi}-\lambda g_{\lambda,\phi}=f_{\lambda,\phi} \quad \text{in } \Omega, \\
&g_{\lambda,\phi}\big|_{+}=0 \quad \text{on } \partial \Omega,
\end{aligned}
\right.
\end{equation*}
since the resolvent is well defined for $\{\lambda:|\lambda|\leq \lambda_0\}$.
Finally, we let
\begin{equation} \label{eq_d2n_1}
u_{\lambda,\phi}:=E\phi +g_{\lambda,\phi}=E\phi +(\mathcal{L}^{ext}-\lambda)^{-1}f_{\lambda,\phi},
\end{equation}
and see that $u_{\lambda,\phi}$ solves \eqref{eq_d2n_def}.

{\color{blue}Step 3.} A direct consequence of \eqref{eq_d2n_1} is the following expression for the Dirichlet-to-Neumann map:
\begin{equation} \label{eq_d2n_2}
\mathcal{T}^{\lambda}:\quad
\phi\mapsto \frac{\partial}{\partial\nu}\Big|_{\partial \mathcal{D}}\Big(E\phi +(\mathcal{L}^{ext}-\lambda)^{-1}f_{\lambda,\phi} \Big) .
\end{equation}
This, together with the analyticity of resolvent $(\mathcal{L}^{ext}-\lambda)^{-1}$ and $f_{\lambda,\phi}$, gives the boundedness and analyticity of $\mathcal{T}^{\lambda}$.

\begin{remark} \label{rmk_explicit_d2n}
The Dirichlet-to-Neumann map $\mathcal{T}^{\lambda}$ can be computed explicitly when the inclusions take some special shapes. We briefly show the idea for the following simple case: there is a single resonator ($d=1$) that is a disk located at the center of the unit cell, the background lattice is square-type, \textit{i.e.}, $\bm{e}_1=(1,0)$, $\bm{e}_2=(0,1)$, and the frequency is fixed at $\lambda=0$. First, define the quasi-periodic Green function associated with the homogeneous space (\textit{cf.} \cite[Section 2.12]{ammari2018mathematical_method})
\begin{equation*}
G^{0,\bm{\kappa}}(\bm{x},\bm{y})
=\left\{
\begin{aligned}
&n^{2}_{e}\sum_{\bm{n}\in \mathbb{Z}^2}\frac{e^{i(2\pi\bm{n}+\bm{\kappa})\cdot (\bm{x}-\bm{y})}}{|2\pi\bm{n}+\bm{\kappa}|^2},\quad \bm{\kappa}\neq 0,\\
&n^{2}_{e}\sum_{\bm{n}\in \mathbb{Z}^2\backslash\{0\}}\frac{e^{i2\pi\bm{n}\cdot (\bm{x}-\bm{y})}}{4\pi^2|\bm{n}|^2},\quad \bm{\kappa}= 0 .
\end{aligned}
\right.
\end{equation*}
Then we set
\begin{equation*}
G^{0,\bm{\kappa}}_{ext}(\bm{x},\bm{y})
=G^{0,\bm{\kappa}}(\bm{x},\bm{y})-G^{0,\bm{\kappa}}(\bm{x},\tilde{\bm{y}}),\quad \bm{x},\bm{y}\in Y,
\end{equation*}
where $\tilde{\bm{y}}$ is the reflection image of $\bm{y}$ with respect to $\partial \mathcal{D}_{(1,1)}$ (see the definition of reflection image in \cite[Section 2.2]{evans2022pde}). By the reflection principle, it is easy to check that $G^{0,\bm{\kappa}}_{ext}(\bm{x},\bm{y})$ satisfies the Dirichlet boundary condition on the surface of the resonators $\partial\mathcal{D}$ and the quasi-periodic boundary conditions on $\partial Y$. Finally, we take the inverse Floquet transform and obtain the Green function associated with the exterior structure:
\begin{equation*}
G^{0}_{ext}(\bm{x},\bm{y})
=\int_{Y^{*}}G^{0,\bm{\kappa}}_{ext}(\bm{x},\bm{y})d\bm{\kappa},
\end{equation*}
\textit{i.e.}, $G^{0}_{ext}$ is a parametrix of the operator $\mathcal{L}^{ext}$. Integrating by parts, we can express the solution $u$ to problem \eqref{eq_d2n_def} when $\lambda=0$ in the following double-layer potential form:
\begin{equation*}
u=u_{\phi}=\int_{\partial \mathcal{D}}\frac{\partial G^{0}_{ext}}{\partial \nu_{\bm{y}}}(\bm{x},\bm{y})\phi(\bm{y}) ds(\bm{y}),\quad \bm{x}\in \Omega .
\end{equation*}
Hence,  $\mathcal{T}^{0}$ can be expressed as 
\begin{equation*}
\mathcal{T}^{0}[\phi]=\frac{\partial u_{\phi}}{\partial \nu_{\bm{x}}}\Big|_{\partial \mathcal{D}}
=\frac{\partial}{\partial \nu_{\bm{x}}}\Big|_{\partial \mathcal{D}}\int_{\partial \mathcal{D}}\frac{\partial G^{0}_{ext}}{\partial \nu_{\bm{y}}}(\bm{x},\bm{y})\phi(\bm{y}) ds(\bm{y}) .
\end{equation*}
Note that, using the expansion of the Green function in \cite[Lemma 2.91]{ammari2018mathematical_method}, the above calculation can be generalized to the case of any nonzero $\lambda$ that is sufficiently close to zero.
\end{remark}

\subsection{Proof of Proposition \ref{prop_cap_operator}}

The periodicity of $\mathcal{C}$ follows directly from \eqref{eq_d2n_2} and the fact that both the resolvent $(\mathcal{L}^{ext}-\lambda)^{-1}$ and the lifting operator $E$ are periodic. The real-valuedness of $\mathcal{C}$ follows from the fact that $\mathcal{T}^{0}[\mathbbm{1}_{\partial \mathcal{D}_{\bm{m}}}^{(i)}]$ is a real function (since the solution $u$ to \eqref{eq_d2n_def} is apparently real valued when $\phi=\mathbbm{1}_{\partial \mathcal{D}_{\bm{m}}}^{(i)}$ and $\lambda=0$). The exponential decay follows from the decay of the resolvent $(\mathcal{L}^{ext}-\lambda)^{-1}$ since $\lambda$ is separated from the spectrum of $\mathcal{L}^{ext}$. Indeed, using a standard Combes-Thomas argument as in \cite[Proposition 3.1]{qiu2025bulk}, one can show that
\begin{equation*}
\Big|\nabla(\mathcal{L}^{ext}-\lambda)^{-1}(\bm{x},\bm{y}) \Big|
\leq C_1 e^{-\beta_1 |\bm{x}-\bm{y}|},\quad \bm{n}\neq \bm{m},\quad \forall \bm{x}\in \partial \mathcal{D}_{\bm{n}}^{(i)},\bm{y}\in \partial \mathcal{D}_{\bm{m}}^{(j)} ,
\end{equation*}
where $C,\beta_1>0$ depend only on the geometry of the resonators $\mathcal{D}$ and the isolation distance from $\lambda$ to the spectrum $\text{Spec}(\mathcal{L}^{ext})$. Hence,
\begin{equation*}
\begin{aligned}
\Big|\frac{\partial}{\partial\nu}\Big|_{\partial \mathcal{D}_{\bm{n}}^{(i)}}\Big((\mathcal{L}^{ext}-\lambda)^{-1}f_{\lambda,\mathbbm{1}_{\partial \mathcal{D}_{\bm{m}}^{(i)}}} \Big)\Big|
&\leq C_1 \sup_{\bm{x}\in \partial \mathcal{D}_{\bm{n}}^{(i)},\bm{y}\in \partial \mathcal{D}_{\bm{m}}^{(j)}} e^{-\beta_1 |\bm{x}-\bm{y}|} \int_{\partial \mathcal{D}}\Big|f_{\lambda,\mathbbm{1}_{\partial \mathcal{D}_{\bm{m}}^{(i)}}}\Big| \\
&\leq C_2 e^{-\beta_2 |\bm{n}-\bm{m}|} .
\end{aligned}
\end{equation*}
On the other hand, the localness of the lifting operator $E$ constructed in Section \ref{sec:3} implies that, for $\bm{n}\neq\bm{m}$, $\text{supp}(E\mathbbm{1}_{\partial \mathcal{D}_{\bm{m}}^{(j)}})\cap \overline{\mathcal{D}_{\bm{n}}^{(i)}}=\emptyset$, and hence
\begin{equation*}
\frac{\partial}{\partial\nu}\Big|_{\partial \mathcal{D}_{\bm{n}}^{(i)}}\Big(E\mathbbm{1}_{\partial \mathcal{D}_{\bm{m}}^{(j)}}\Big)=0,\quad (\bm{n}\neq\bm{m}).
\end{equation*}
Then we deduce from \eqref{def:capop} and \eqref{eq_d2n_2} the exponential decay of the capacitance operator $\mathcal{C}$.

\section{Tight-binding approximation: Proof of Theorem \ref{thm_tight_binding_approx}}
\label{sec:4}

The idea of proof is similar to \cite[Theorem 4.5]{ammari2025nonlinear_resonance}, but generalizes it to the infinite periodic system. The first step is to develop an equivalent formulation of problem \eqref{eq_pde_station_weak}. This is achieved by introducing an auxiliary functional, where adding a volume integral makes it invertible in the following sense.
\begin{proposition}\label{prop_1st_aux_func}
Define the following functional on $H^{1}(\mathcal{D})\times H^{1}(\mathcal{D})$:
\begin{equation} \label{eq_1st_aux_form}
\mathfrak{a}^{aux,I}_{\lambda,\delta}(u,v):= (\nabla u,\nabla v)_{\mathcal{D}}+\sum_{\substack{1\leq i\leq d \\ \bm{n}\in\mathbb{Z}^2}}\int_{\mathcal{D}_{\bm{n}}^{(i)}}u\int_{\mathcal{D}_{\bm{n}}^{(i)}}\overline{v}-\lambda ( u, v)_{\mathcal{D}}-\lambda \sigma ( |u|^2 u, v)_{\mathcal{D}} -\delta \langle \mathcal{T}^{\lambda}[u\big|_{\partial \mathcal{D}}], v\rangle .
\end{equation}
There exist $\delta_1,\lambda_1,M_{1}>0$, which depend only on the geometry of resonators $\mathcal{D}$ and the nonlinearity $\sigma$, such that for $\delta<\delta_1$, $\lambda<\lambda_1$ and $\|f\|_{-1,\mathcal{D}}\leq M_{1}$, there exists $u_{f}\in H^1(\mathcal{D})$ such that $\mathfrak{a}^{aux,I}_{\lambda,\delta}(u_{f},v)=(f,v)_{\mathcal{D}}$ for all $v\in H^1(\mathcal{D})$. In particular, $u_{f}$ satisfies the estimate
\begin{equation} \label{eq_1st_aux_sol_estimate}
\|u_f\|_{1,\mathcal{D}}\leq C\|f\|_{-1,\mathcal{D}}, \quad \text{with C depending only on $\sigma,\lambda_1,\delta_1,M_1$.}
\end{equation}
Moreover, under the above conditions, the solution to $\mathfrak{a}^{aux,I}_{\lambda,\delta}(u_{f},v)=(f,v)_{\mathcal{D}}$ that satisfies the bound \eqref{eq_1st_aux_sol_estimate} is unique.
\end{proposition}
\begin{proof}
See Section \ref{sec_4_1}.
\end{proof}
\begin{remark}
We will not stress the dependence of constants on the geometry of the resonators $\mathcal{D}$ and nonlinearity $\sigma$ in the sequel, as they are fixed throughout this paper.
\end{remark}
With this auxiliary form in hand, it immediately follows that the original eigenvalue problem \eqref{eq_pde_station_weak} is transformed to the following version.
\begin{proposition} \label{prop_exact_equivalence}
Let $\delta_1,\lambda_1,M_{1}>0$ be the constants fixed in Proposition \ref{prop_1st_aux_func}. Then, for $\delta<\delta_1$, the following statements are equivalent:
\begin{itemize}
    \item[i)] the eigenvalue problem \eqref{eq_pde_station_weak} has a solution $(u,\lambda)$ with $\lambda<\lambda_1$ and $\|u\|_{L^2(\mathcal{D})}\leq M_{1}$;
    \item[ii)] the solution $u$ to $\mathfrak{a}^{aux,I}_{\lambda,\delta}(u,v)=(f,v)_{\mathcal{D}}$ with $f=\sum_{\substack{1\leq i\leq d \\ \bm{n}\in\mathbb{Z}^2}}a_{\bm{n}}^{(i)}\mathbbm{1}_{\mathcal{D}_{\bm{n}}^{(i)}}\neq 0$ ($\|a\|_{\ell^2(\mathbb{Z}^2;\mathbb{C}^d)}\leq M_{1}$) satisfies $\int_{\mathcal{D}}u\mathbbm{1}_{\mathcal{D}_{\bm{n}}^{(i)}}=a_{\bm{n}}^{(i)}$.
\end{itemize}
\end{proposition}
This new formulation of \eqref{eq_pde_station_weak} has the advantage of transforming the (continuous) problem originally posed on $L^2(\mathbb{R}^2)$ into a (discrete) problem posed on $\ell^2(\mathbb{Z}^2;\mathbb{C}^{d})$. We have a clear asymptotic expansion of the solution to the new formulation.
\begin{proposition} \label{prop_ansatz_aux_sol}
Let $\delta_1,\lambda_1,M_{1}>0$ be the constants fixed in Proposition \ref{prop_1st_aux_func}. Then, for $\delta<\delta_1$, $\lambda<\lambda_1$ and $f=\sum_{\substack{1\leq i\leq d \\ \bm{n}\in\mathbb{Z}^2}}a_{\bm{n}}^{(i)}\mathbbm{1}_{\mathcal{D}_{\bm{n}}^{(i)}}$ with $\|a\|_{\ell^2(\mathbb{Z}^2;\mathbb{C}^{d})}\leq M_{1}$, the solution to $$\mathfrak{a}^{aux,I}_{\lambda,\delta}(u_{f},v)=(f,v)_{\mathcal{D}}$$ admits the following expansion:
\begin{equation} \label{eq_ansatz_aux_sol_1}
u_{f}=u^{(0,0)}_{a}+\lambda u^{(1,0)}_{a} +\delta u^{(0,1)}_{a} +r_{a},
\end{equation}
where
\begin{equation} \label{eq_ansatz_aux_sol_2}
u^{(0,0)}_{a}:=\sum_{\substack{1\leq i\leq d \\ \bm{n}\in\mathbb{Z}^2}}a_{\bm{n}}^{(i)}\mathbbm{1}_{\mathcal{D}_{\bm{n}}^{(i)}} ,\quad
u^{(1,0)}_{a}:=\sum_{\substack{1\leq i\leq d \\ \bm{n}\in\mathbb{Z}^2}}(a_{\bm{n}}^{(i)}+\sigma |a_{\bm{n}}^{(i)}|^2a_{\bm{n}}^{(i)})\mathbbm{1}_{\mathcal{D}_{\bm{n}}^{(i)}},
\end{equation}
$u^{(0,1)}_{a}$ is the unique solution to the following linear problem:
\begin{equation} \label{eq_2nd_aux_form}
\mathfrak{a}^{aux,II}(u^{(0,1)}_{a},v):=(\nabla u^{(0,1)}_{a},\nabla v)_{\mathcal{D}}+\sum_{\substack{1\leq i\leq d \\ \bm{n}\in\mathbb{Z}^2}}\int_{\mathcal{D}_{\bm{n}}^{(i)}}u^{(0,1)}_{a}\int_{\mathcal{D}_{\bm{n}}^{(i)}}\overline{v} =        \langle \mathcal{T}^{0}[u^{(0,0)}_{a}\big|_{\partial \mathcal{D}}], v\rangle , 
\end{equation}
and satisfies
\begin{equation} \label{eq_u01_boundary_int}
\int_{\mathcal{D}_{\bm{n}}^{(i)}}u^{(0,1)}_{a}=\sum_{\substack{1\leq j\leq d \\ \bm{n}\in\mathbb{Z}^2}}a_{\bm{m}}^{(j)}\int_{\partial \mathcal{D}}\mathcal{T}^{0}[\mathbbm{1}_{\partial \mathcal{D}_{\bm{m}}^{(j)}}]\cdot \mathbbm{1}_{\partial \mathcal{D}_{\bm{n}}^{(i)}} .
\end{equation}
The remainder $r_{a}$ solves the following nonlinear problem:
\begin{equation} \label{eq_3rd_aux_form}
\begin{aligned}
\mathfrak{a}^{aux,III}_{a,\lambda,\delta}(r_{a},v)
:=& (\nabla r_{a},\nabla v)_{\mathcal{D}}+\sum_{\substack{1\leq i\leq d \\ \bm{n}\in\mathbb{Z}^2}}\int_{\mathcal{D}_{\bm{n}}^{(i)}}r_{a}\int_{\mathcal{D}_{\bm{n}}^{(i)}}\overline{v}-\lambda(r_{a},v)_{\mathcal{D}}-\lambda\sigma (|u^{(0,0)}_{a}|^2 r_{a},v)_{\mathcal{D}} \\
&-2\lambda\sigma \left(\Re(r_{a}\overline{u^{(0,0)}_{a}})r_{a},v\right)_{\mathcal{D}}-2\lambda\sigma \left(\Re((\lambda u^{(1,0)}_{a}+\delta u^{(0,1)}_{a})\overline{u^{(0,0)}_{a}})r_{a},v\right)_{\mathcal{D}} \\
&-2\lambda\sigma \left((u^{(0,0)}_{a}+\lambda u^{(1,0)}_{a}+\delta u^{(0,1)}_{a})\Re(\overline{u^{(0,0)}_{a}}r_{a}),v\right)_{\mathcal{D}} -\lambda\sigma (|\lambda u^{(1,0)}_{a}+\delta u^{(0,1)}_{a}|^2 r_{a},v)_{\mathcal{D}} \\
&-2\lambda\sigma\left( (u^{(0,0)}_{a}+\lambda u^{(1,0)}_{a}+\delta u^{(0,1)}_{a})\Re( (\overline{\lambda u^{(1,0)}_{a}+\delta u^{(0,1)}_{a}})r_{a} ),v \right)_{\mathcal{D}} \\
&-\lambda\sigma \left( |r_{a}|^2(u^{(0,0)}_{a}+\lambda u^{(1,0)}_{a}+\delta u^{(0,1)}_{a}),v \right)_{\mathcal{D}} -\lambda\sigma (|r_{a}|^2r_{a},v)_{\mathcal{D}} \\
&-2\lambda\sigma \left( r_{a}\Re( (\overline{\lambda u^{(1,0)}_{a}+\delta u^{(0,1)}_{a}})r_{a},v ) \right)_{\mathcal{D}} -\delta\langle \mathcal{T}^{\lambda}[r_{a}\big|_{\partial \mathcal{D}}],v \rangle \\
=& (f_{a,\lambda,\delta},v)_{\mathcal{D}} 
\end{aligned}
\end{equation}
with
\begin{equation} \label{eq_ansatz_aux_sol_3}
\begin{aligned}
f_{a,\lambda,\delta}
:=&\lambda^2 u^{(1,0)}_{a}+\lambda\delta u^{(0,1)}_{a} + \lambda\sigma |u^{(0,0)}_{a}|^2(\lambda u^{(1,0)}_{a}+\delta u^{(0,1)}_{a}) \\
&+ 2\lambda\sigma \Re\left(\overline{u^{(0,0)}_{a}}(\lambda u^{(1,0)}_{a}+\delta u^{(0,1)}_{a}) \right)\cdot (u^{(0,0)}_{a}+\lambda u^{(1,0)}_{a}+\delta u^{(0,1)}_{a}) \\
&+\lambda\sigma |\lambda u^{(1,0)}_{a}+\delta u^{(0,1)}_{a}|^2 (u^{(0,0)}_{a}+\lambda u^{(1,0)}_{a}+\delta u^{(0,1)}_{a}) \\
&+\delta (\mathcal{T}^{\lambda}-\mathcal{T}^{0})[u^{(0,0)}_{a}\big|_{\partial \mathcal{D}}]\mathbbm{1}_{\partial \mathcal{D}}
+\delta\mathcal{T}^{\lambda}[\lambda u^{(1,0)}_{a}\big|_{\partial \mathcal{D}}+\delta u^{(0,1)}_{a}\big|_{\partial \mathcal{D}}]\mathbbm{1}_{\partial \mathcal{D}}.
\end{aligned}
\end{equation}
The nonlinear form $\mathfrak{a}^{aux,III}_{a,\lambda,\delta}$ is invertible in the following sense: there exist $\delta_2,\lambda_2,M_{2},$ and $M_{3}>0$ such that for $\delta<\delta_2$, $\lambda<\lambda_2$ and $\|a\|_{\ell^{2}(\mathbb{Z}^2;\mathbb{C}^{d})}<M_2$, the nonlinear problem $$\mathfrak{a}^{aux,III}_{a,\lambda,\delta}(u,v)=(f,v)_{\mathcal{D}}$$ has a solution $u=u_{f}\in H^1(\mathcal{D})$ for any $f\in (H^1(\mathcal{D}))^{*}$ with $\|f\|_{-1,\mathcal{D}}< M_3$. In particular, $u_{f}$ satisfies the estimate
\begin{equation} \label{eq_ansatz_aux_sol_4}
\|u_{f}\|_{1,\mathcal{D}}\leq C\|f\|_{-1,\mathcal{D}} \quad \text{with C depending only on $\lambda_2,\delta_2,M_2,M_3$.}
\end{equation}
Moreover, under the above conditions, the solution to $\mathfrak{a}^{aux,III}_{a,\lambda,\delta}(u,v)=(f,v)_{\mathcal{D}}$ that satisfies the bound \eqref{eq_ansatz_aux_sol_4} is unique.
\end{proposition}
\begin{proof}
    See Section \ref{sec_4_2}.
\end{proof}
\begin{remark}
The boundary terms in \eqref{eq_ansatz_aux_sol_3} apply on $v\in H^1(\mathcal{D})$ in the sense of trace: $\big(u\mathbbm{1}_{\partial \mathcal{D}},v\big)_{\mathcal{D}}:=\langle u,v\big|_{\partial \mathcal{D}} \rangle$.
\end{remark}

Using Propositions \ref{prop_1st_aux_func}-\ref{prop_ansatz_aux_sol}, we are now able to prove Theorem \ref{thm_tight_binding_approx}. Choose the constants in Theorem \ref{thm_tight_binding_approx} as follows: $$\delta_0:=\min\{\delta_1,\delta_2\}, \quad \lambda_0^{cont}:=\min\{\lambda_1,\lambda_2\}, \quad \text{and } \quad M_0^{cont}:=\min\{M_1,M_2,M_3\}.$$ The constants $\lambda_0^{disc}$ and $M_0^{disc}$ will be chosen in the sequel; see \eqref{eq_tight_bind_proof_11}.

{\color{blue}Step 1.} First, suppose that $u\in L^2(\mathbb{R}^2)$ is a continuous soliton with eigenvalue $\lambda^{cont}<\lambda_0^{cont}$ and satisfies $\|u\|_{L^2(\mathbb{R}^2)}\leq M_0$. We now verify that $a=(a_{\bm{n}}^{(i)})$ with $a_{\bm{n}}^{(i)}:=\int_{\mathcal{D}_{\bm{n}}^{(i)}}u$ is indeed an approximate discrete soliton. In fact, by the equivalent formulation in Proposition \ref{prop_exact_equivalence}, we see that the solution $u_f$ to $\mathfrak{a}^{aux,I}_{\lambda,\delta}(u,v)=(f,v)_{\mathcal{D}}$ with $f=\sum_{\substack{1\leq i\leq d \\ \bm{n}\in\mathbb{Z}^2}}a_{\bm{n}}^{(i)}\mathbbm{1}_{\mathcal{D}_{\bm{n}}^{(i)}}\neq 0$ satisfies:
\begin{equation*}
\int_{\mathcal{D}_{\bm{n}}^{(i)}}u_{f}=a_{\bm{n}}^{(i)}.
\end{equation*}
Expanding $u_{f}$ with \eqref{eq_ansatz_aux_sol_1} in Proposition \ref{prop_ansatz_aux_sol}, we derive
\begin{equation*}
\begin{aligned}
0&=-a_{\bm{n}}^{(i)}+\int_{\mathcal{D}_{\bm{n}}^{(i)}}u_{f}
=-a_{\bm{n}}^{(i)}+\int_{\mathcal{D}_{\bm{n}}^{(i)}}(u^{(0,0)}_{a}+\lambda u^{(1,0)}_{a} +\delta u^{(0,1)}_{a} +r_{a}) \\
&=-a_{\bm{n}}^{(i)}+a_{\bm{n}}^{(i)}+\lambda(a_{\bm{n}}^{(i)}+\sigma|a_{\bm{n}}^{(i)}|^2a_{\bm{n}}^{(i)})+\delta\sum_{\substack{1\leq j\leq d \\ \bm{m}\in\mathbb{Z}^2}}a_{\bm{m}}^{(j)} \int_{\partial \mathcal{D}_{\bm{n}}^{(i)}}\mathcal{T}^{0}[\mathbbm{1}_{\partial \mathcal{D}_{\bm{m}}^{(j)}}] + \int_{\partial \mathcal{D}_{\bm{n}}^{(i)}}r_{a} \\
&=\lambda(a+\sigma|a|^2a)_{\bm{n}}^{(i)}-\delta(\mathcal{C}a)_{\bm{n}}^{(i)}+(r^{disc}_{a})_{\bm{n}}^{(i)}
\end{aligned}
\end{equation*}
with $$(r^{disc}_{a})_{\bm{n}}^{(i)}:=\int_{\partial \mathcal{D}_{\bm{n}}^{(i)}}r_{a}$$ This implies
\begin{equation} \label{eq_tight_bind_proof_1}
\mathfrak{a}^{disc}_{\lambda/\delta }(a,b)=\frac{1}{\delta}(r^{disc}_{a},b)_{\ell^2(\mathbb{Z}^2;\mathbb{C}^{d})} \quad (\forall b\in \ell^2(\mathbb{Z}^2;\mathbb{C}^{d})).
\end{equation}
Hence, $a\in \ell^2(\mathbb{Z}^2;\mathbb{C}^d)$ is indeed an approximate discrete soliton if the remainder $r_{a}^{disc}$ is small. We claim that the remainder is bounded by
\begin{equation} \label{eq_tight_bind_proof_2}
\|r^{disc}_{a}\|_{\ell^2(\mathbb{Z}^2;\mathbb{C}^{d})}=\mathcal{O}(\lambda^2+\delta^2).
\end{equation}
Then \eqref{eq_tight_bind_proof_1} and \eqref{eq_tight_bind_proof_2} conclude the proof of the first part of Theorem \ref{thm_tight_binding_approx}.

\begin{proof}[Proof of \eqref{eq_tight_bind_proof_2}]
This follows from expansion \eqref{eq_ansatz_aux_sol_3} and estimate \eqref{eq_ansatz_aux_sol_4}. In fact, we successively apply the Schwarz inequality in \eqref{eq_ansatz_aux_sol_3} and obtain that
\begin{equation} \label{eq_tight_bind_proof_3}
\begin{aligned}
\|f_{a,\lambda,\delta}\|_{-1,\mathcal{D}}
&\leq \lambda^2 \|u^{(1,0)}_{a}\|_{0,\mathcal{D}}+\lambda\delta \|u^{(0,1)}_{a}\|_{0,\mathcal{D}} + \lambda|\sigma| \||u^{(0,0)}_{a}|^2(\lambda u^{(1,0)}_{a}+\delta u^{(0,1)}_{a})\|_{0,\mathcal{D}} \\
&\quad+ 2\lambda|\sigma| \Big\|\Re\left(\overline{u^{(0,0)}_{a}}(\lambda u^{(1,0)}_{a}+\delta u^{(0,1)}_{a}) \right)\cdot (u^{(0,0)}_{a}+\lambda u^{(1,0)}_{a}+\delta u^{(0,1)}_{a})\Big\|_{0,\mathcal{D}} \\
&\quad +\lambda|\sigma| \big\| |\lambda u^{(1,0)}_{a}+\delta u^{(0,1)}_{a}|^2 (u^{(0,0)}_{a}+\lambda u^{(1,0)}_{a}+\delta u^{(0,1)}_{a}) \big\|_{0,\mathcal{D}} \\
&\quad +\delta\|\mathcal{T}^{\lambda}-\mathcal{T}^{0}\|_{\mathcal{B}(H^{\frac{1}{2}}(\partial \mathcal{D}),H^{-\frac{1}{2}}(\partial \mathcal{D}))} \|u^{(0,0)}_{a}\|_{1,\mathcal{D}}
+\lambda\delta\|\mathcal{T}^{\lambda}\|_{\mathcal{B}(H^{\frac{1}{2}}(\partial \mathcal{D}),H^{-\frac{1}{2}}(\partial \mathcal{D}))}\|u^{(1,0)}_{a}\|_{1,\mathcal{D}} \\
&\quad + \delta^2\|\mathcal{T}^{\lambda}\|_{\mathcal{B}(H^{\frac{1}{2}}(\partial \mathcal{D}),H^{-\frac{1}{2}}(\partial \mathcal{D}))}\|u^{(0,1)}_{a}\|_{1,\mathcal{D}} \\
&\leq C(\lambda^2+\delta^2)\Big\{
\|u^{(0,0)}_{a}\|_{1,\mathcal{D}}+\|u^{(1,0)}_{a}\|_{1,\mathcal{D}}+ \|u^{(0,1)}_{a}\|_{1,\mathcal{D}}\\
&\quad\quad\quad+\|u^{(0,0)}_{a}\|_{L^4(\mathcal{D})}^{2}(\|u^{(1,0)}_{a}\|_{0,\mathcal{D}}+\|u^{(0,1)}_{a}\|_{0,\mathcal{D}}) \\
&\quad\quad\quad +\|u^{(0,0)}_{a}\|_{L^4(\mathcal{D})}(\|u^{(1,0)}_{a}\|_{L^4(\mathcal{D})}+\|u^{(0,1)}_{a}\|_{L^4(\mathcal{D})})(\|u^{(0,0)}_{a}\|_{0,\mathcal{D}}+\|u^{(1,0)}_{a}\|_{0,\mathcal{D}}+\|u^{(0,1)}_{a}\|_{0,\mathcal{D}}) \\
&\quad\quad\quad +(\|u^{(1,0)}_{a}\|_{L^4(\mathcal{D})}+\|u^{(0,1)}_{a}\|_{L^4(\mathcal{D})})^2(\|u^{(0,0)}_{a}\|_{0,\mathcal{D}}+\|u^{(1,0)}_{a}\|_{0,\mathcal{D}}+\|u^{(0,1)}_{a}\|_{0,\mathcal{D}}) \Big\},
\end{aligned}
\end{equation}
where $C>0$ is independent of $\lambda$ and $\delta$. Note that expansion \eqref{eq_ansatz_aux_sol_2} indicates that $u^{(0,0)}_{a}$ and $u^{(1,0)}_{a}$ are constant-valued in each resonator; hence their $H^1(\mathcal{D})$ norms are equal to their $L^2(\mathcal{D})$ norms, which are estimated as
\begin{equation} \label{eq_tight_bind_proof_4}
\|u^{(0,0)}_{a}\|_{1,\mathcal{D}}=\|u^{(0,0)}_{a}\|_{0,\mathcal{D}}=\|a\|_{\ell^2(\mathbb{Z}^2;\mathbb{C}^{d})}=\Big(\sum_{\substack{1\leq i\leq d \\ \bm{n}\in\mathbb{Z}^2}}\big|\int_{\mathcal{D}_{\bm{n}}^{(i)}}u\big|^2 \Big)^{1/2}\leq \|u\|_{0,\mathcal{D}},
\end{equation}
and
\begin{equation} \label{eq_tight_bind_proof_6}
\begin{aligned}
\|u^{(1,0)}_{a}\|_{1,\mathcal{D}}&=\|u^{(1,0)}_{a}\|_{0,\mathcal{D}}
\leq\|a\|_{\ell^2(\mathbb{Z}^2)}+|\sigma|\||a|^2 a\|_{\ell^2(\mathbb{Z}^2;\mathbb{C}^{d})}
=\|a\|_{\ell^2(\mathbb{Z}^2;\mathbb{C}^{d})}+|\sigma|\| a\|_{\ell^{6}(\mathbb{Z}^2;\mathbb{C}^{d})}^{3} \\
&\leq C\|a\|_{\ell^2(\mathbb{Z}^2;\mathbb{C}^{d})}(1+\|a\|_{\ell^2(\mathbb{Z}^2;\mathbb{C}^{d})})^2
\leq C\|u\|_{0,\mathcal{D}}(1+\|u\|_{0,\mathcal{D}})^2,
\end{aligned}
\end{equation}
where the standard estimate $$\|a\|_{\ell^k(\mathbb{Z}^2;\mathbb{C}^{d})}\leq \|a\|_{\ell^2(\mathbb{Z}^2;\mathbb{C}^{d})} \quad (\forall k\geq 2)$$ is applied to derive the second inequality in \eqref{eq_tight_bind_proof_6}. Similarly, their $L^4(\mathcal{D})-$norms are estimated as
\begin{equation} \label{eq_tight_bind_proof_5}
\|u^{(0,0)}_{a}\|_{L^4(\mathcal{D})}=\|a\|_{\ell^4(\mathbb{Z}^2;\mathbb{C}^{d})}\leq \|a\|_{\ell^2(\mathbb{Z}^2;\mathbb{C}^{d})} \leq \|u\|_{0,\mathcal{D}} ,
\end{equation}
and
\begin{equation} \label{eq_tight_bind_proof_7}
\begin{aligned}
\|u^{(1,0)}_{a}\|_{L^4(\mathcal{D})}
&\leq\|a\|_{\ell^4(\mathbb{Z}^2;\mathbb{C}^{d})}+|\sigma|\||a|^2 a\|_{\ell^4(\mathbb{Z}^2;\mathbb{C}^{d})} \\
&\leq C\|a\|_{\ell^2(\mathbb{Z}^2;\mathbb{C}^{d})}(1+\|a\|_{\ell^2(\mathbb{Z}^2;\mathbb{C}^{d})})^2
\leq C\|u\|_{0,\mathcal{D}}(1+\|u\|_{0,\mathcal{D}})^2 .
\end{aligned}
\end{equation}
On the other hand, the norm of $u^{(0,1)}_{a}$ is estimated as follows:
\begin{equation} \label{eq_tight_bind_proof_8}
\begin{aligned}
\|u^{(0,1)}_{a}\|_{0,\mathcal{D}}
&\leq  \|u^{(0,1)}_{a}\|_{1,\mathcal{D}} \overset{(i)}{\leq} C \|\mathcal{T}^{\lambda}[u^{(0,0)}_{a}\big|_{\partial \mathcal{D}}]\|_{H^{1/2}(\partial \mathcal{D})} \\
&\overset{(ii)}{\leq} C \|u^{(0,0)}_{a}\|_{1, \mathcal{D}}
= C\|u^{(0,0)}_{a}\|_{0, \mathcal{D}}
\leq C\|a\|_{\ell^2(\mathbb{Z}^2;\mathbb{C}^{d})} \leq  C\|u\|_{0,\mathcal{D}}.
\end{aligned}
\end{equation}
Here, to prove inequality (i), we recall the fact that $u^{(0,1)}_{a}$ solves \eqref{eq_2nd_aux_form} and that the sesquilinear form $$\mathfrak{a}^{aux,II}(u,v)=(\nabla u,\nabla v)_{\mathcal{D}}+\sum_{\substack{1\leq i\leq d \\ \bm{n}\in\mathbb{Z}^2}}\int_{\mathcal{D}_{\bm{n}}^{(i)}}u\int_{\mathcal{D}_{\bm{n}}^{(i)}}\overline{v}$$ is continuously coercive (see the proof of Lemma \ref{lem_4th_form_invert} for details). Inequality (ii) follows from the boundedness of the DtN map $\mathcal{T}^{\lambda}$ (Proposition \ref{prop_d2n_map}). We also obtain the $L^4$ estimate of $u^{(0,1)}_{a}$ by \eqref{eq_tight_bind_proof_8} and the Sobolev embedding (for unbounded domains) $H^{1}(\mathcal{D})\subset L^{4}(\mathcal{D})$:
\begin{equation} \label{eq_tight_bind_proof_9}
\begin{aligned}
\|u^{(0,1)}_{a}\|_{L^4(\mathcal{D})}\leq 
C \|u^{(0,1)}_{a}\|_{1,\mathcal{D}} \leq C\|a\|_{\ell^2(\mathbb{Z}^2;\mathbb{C}^{d})} \leq  C\|u\|_{0,\mathcal{D}} .
\end{aligned}
\end{equation}
With \eqref{eq_tight_bind_proof_4}-\eqref{eq_tight_bind_proof_9}, \eqref{eq_tight_bind_proof_3} is estimated as
\begin{equation} \label{eq_tight_bind_proof_10}
\begin{aligned}
\|f_{a,\lambda,\delta}\|_{-1,\mathcal{D}}
\leq C(\lambda^2+\delta^2)\Big\{\|u\|_{0,\mathcal{D}}(1+\|u\|_{0,\mathcal{D}})^2+\|u\|_{0,\mathcal{D}}^3(1+\|u\|_{0,\mathcal{D}})^6 \Big\} .
\end{aligned}
\end{equation}
This estimate, together with the assumption $\|u\|_{1,\mathcal{D}}\leq M_{0}^{cont}$, implies that one can estimate $\|r_{a}\|_{0,\mathcal{D}}$ using \eqref{eq_ansatz_aux_sol_4} if $\lambda$ and $\delta$ are sufficiently small (such that the right side of \eqref{eq_tight_bind_proof_10} is smaller than $M_3$ and hence the conditions in Proposition \ref{prop_ansatz_aux_sol} are fulfilled):
\begin{equation*}
\|r_{a}\|_{0,\mathcal{D}} \leq C\|f_{a,\lambda,\delta}\|_{-1,\mathcal{D}}=\mathcal{O}(\lambda^2+\delta^2).
\end{equation*}
Consequently,
\begin{equation*}
\|r^{disc}_{a}\|_{\ell^2(\mathbb{Z}^2;\mathbb{C}^{d})}\leq \|r_{a}\|_{0,\mathcal{D}} =\mathcal{O}(\lambda^2+\delta^2),
\end{equation*}
which concludes the proof of \eqref{eq_tight_bind_proof_2}.
\end{proof}

{\color{blue}Step 2.} Now we prove the second part of Theorem \ref{thm_tight_binding_approx}. Suppose that $a\in \ell^2(\mathbb{Z}^2;\mathbb{C}^{d})$ is a discrete subwavelength soliton with eigenvalue $\lambda^{disc}<\lambda_0^{disc}$ and satisfies $\|a\|_{\ell^2(\mathbb{Z}^2;\mathbb{C}^{d})}\leq M_0^{disc}$ ($\lambda_0^{disc}$ and $M_0^{disc}$ are to be specified later). We define $\lambda^{cont}:=\delta\lambda^{disc}$ and
\begin{equation} \label{eq_tight_bind_proof_11}
u:=u^{(0,0)}_{a}+\lambda^{cont} u^{(1,0)}_{a} +\delta u^{(0,1)}_{a} +r_{a},
\end{equation}
where $u^{(0,0)}_{a},u^{(1,0)}_{a}$ are introduced in \eqref{eq_ansatz_aux_sol_2}, $u^{(0,1)}_{a}$ and $r_{a}$ are respectively the unique solutions to \eqref{eq_2nd_aux_form} and $$\mathfrak{a}^{aux,III}_{\lambda^{cont},\delta}(r_{a},v)=(f_{a,\lambda^{cont},\delta},v)_{\mathcal{D}}.$$ We next verify that \eqref{eq_tight_bind_proof_11} gives an approximate continuous soliton with eigenvalue $\lambda^{cont}$. 

{\color{blue}Step 2.1.} First, we need to show that $r_a$ is well defined. In fact, it suffices to prove that when the thresholds $\lambda_0^{disc},\delta_0,$ and $M_0^{disc}$ are sufficiently small, the conditions in Proposition \ref{prop_ansatz_aux_sol} are satisfied, which hence guarantees that the following problem has a unique solution
$$
\mathfrak{a}^{aux,III}_{\lambda^{cont},\delta}(r_{a},v)=(f_{a,\lambda^{cont},\delta},v)_{\mathcal{D}}.
$$
The desired constraints on the thresholds are obtained as follows. 
By controlling the right-hand side of \eqref{eq_tight_bind_proof_3} using \eqref{eq_tight_bind_proof_4}-\eqref{eq_tight_bind_proof_9}, we see that
\begin{equation} \label{eq_tight_bind_proof_12}
\begin{aligned}
\|f_{a,\lambda^{cont},\delta}\|_{1,\mathcal{D}}
&\leq C((\lambda^{cont})^2+\delta^2)\Big\{
\|u^{(0,0)}_{a}\|_{1,\mathcal{D}}+\|u^{(1,0)}_{a}\|_{1,\mathcal{D}}+ \|u^{(0,1)}_{a}\|_{1,\mathcal{D}}\\
& +\|u^{(0,0)}_{a}\|_{L^4(\mathcal{D})}^{2}(\|u^{(1,0)}_{a}\|_{0,\mathcal{D}}+\|u^{(0,1)}_{a}\|_{0,\mathcal{D}}) \\
& +\|u^{(0,0)}_{a}\|_{L^4(\mathcal{D})}(\|u^{(1,0)}_{a}\|_{L^4(\mathcal{D})}+\|u^{(0,1)}_{a}\|_{L^4(\mathcal{D})})(\|u^{(0,0)}_{a}\|_{0,\mathcal{D}}+\|u^{(1,0)}_{a}\|_{0,\mathcal{D}}+\|u^{(0,1)}_{a}\|_{0,\mathcal{D}}) \\
& +(\|u^{(1,0)}_{a}\|_{L^4(\mathcal{D})}+\|u^{(0,1)}_{a}\|_{L^4(\mathcal{D})})^2(\|u^{(0,0)}_{a}\|_{0,\mathcal{D}}+\|u^{(1,0)}_{a}\|_{0,\mathcal{D}}+\|u^{(0,1)}_{a}\|_{0,\mathcal{D}}) \Big\} \\
&\leq C\delta^2((\lambda^{disc})^2+1)
\Big\{\|a\|_{\ell^2(\mathbb{Z}^2;\mathbb{C}^{d})}(1+\|a\|_{\ell^2(\mathbb{Z}^2;\mathbb{C}^{d})})^2
+\|a\|_{\ell^2(\mathbb{Z}^2;\mathbb{C}^{d})}^3(1+\|a\|_{\ell^2(\mathbb{Z}^2;\mathbb{C}^{d})})^6 \Big\}  \\
&\leq C\delta^2((\lambda^{disc})^2+1)(1+\|a\|_{\ell^2(\mathbb{Z}^2;\mathbb{C}^{d})})^9 .
\end{aligned}
\end{equation} 
Hence, the appropriate constraints on $\lambda_0^{disc},\delta_0$ and $M_0^{disc}$ are 
\begin{equation} \label{eq_tight_bind_proof_13}
\left\{
\begin{aligned}
&C\delta^2((\lambda^{disc}_{0})^2+1)(1+M_{0}^{disc})^9\leq M_3 , \\
&\delta\cdot \lambda^{disc}_{0} \leq \lambda_{2}, \\
&M^{disc}_{0}\leq M_2.
\end{aligned}
\right.
\end{equation}
Note that we can simply set $\lambda_0^{disc},\delta_0$ and $M_0^{disc}$ to be sufficiently small. Then Proposition \ref{prop_ansatz_aux_sol} guarantees that the function $r_a$ in \eqref{eq_tight_bind_proof_11} is well defined. 

{\color{blue}Step 2.2.} Now we prove that $u$ defined in \eqref{eq_tight_bind_proof_11} is indeed an approximate soliton. By \eqref{eq_ansatz_aux_sol_2} and \eqref{eq_2nd_aux_form}, we can directly check that the function $u$ defined in \eqref{eq_tight_bind_proof_11} actually solves the second problem $$\mathfrak{a}^{aux,I}_{\lambda^{c},\delta}(u,v)=(f,v)_{\mathcal{D}}$$ in Proposition \ref{prop_exact_equivalence} with $$f=\sum_{\substack{1\leq i\leq d \\ \bm{n}\in\mathbb{Z}^2}}a_{\bm{n}}^{(i)}\mathbbm{1}_{\mathcal{D}_{\bm{n}}^{(i)}}.$$ However, $u$ may not solve the original eigenvalue problem \eqref{eq_pde_station_weak} because the condition $$\int_{\mathcal{D}}u\mathbbm{1}_{\mathcal{D}_{\bm{n}}^{(i)}}=a_{\bm{n}}^{(i)}$$ is not necessarily fulfilled due to the remainder $r_a$. Nonetheless, this error is small: in fact, by \eqref{eq_tight_bind_proof_11}, \eqref{eq_ansatz_aux_sol_2} and \eqref{eq_u01_boundary_int}, the error is totally contributed by the remainder $r_a$
\begin{equation*}
\int_{\mathcal{D}_{\bm{n}}^{(i)}}u-a_{\bm{n}}^{(i)}=\int_{\mathcal{D}_{\bm{n}}^{(i)}}r_{a} .
\end{equation*}
Hence, we see that
\begin{equation*}
\begin{aligned}
\mathfrak{a}^{aux,I}_{\lambda^{cont},\delta}(u,v)=(\sum_{\substack{1\leq i\leq d \\ \bm{n}\in\mathbb{Z}^2}}a_{\bm{n}}^{(i)}\mathbbm{1}_{\mathcal{D}_{\bm{n}}^{(i)}},v)_{\mathcal{D}}
=\sum_{\substack{1\leq i\leq d \\ \bm{n}\in\mathbb{Z}^2}}\int_{\mathcal{D}_{\bm{n}}^{(i)}}u\int_{\mathcal{D}_{\bm{n}}^{(i)}}\overline{v}-(r_a,v)_{\mathcal{D}},
\end{aligned}
\end{equation*}
which implies that $\mathfrak{a}^{cont}_{\lambda^{cont},\delta}(u_a,v)
=-(r_a,v)_{\mathcal{D}}$ and hence,
\begin{equation*}
\sup_{\|v\|_{H^1(\mathcal{D})}=1}|\mathfrak{a}^{cont}_{\lambda^{cont},\delta}(u_a,v)|\leq \|r_a\|_{0,\mathcal{D}}\overset{\eqref{eq_tight_bind_proof_12},\eqref{eq_ansatz_aux_sol_4}}{=\joinrel=}\mathcal{O}(\delta^2).
\end{equation*}
Note that the remainder $\mathfrak{e}$ in Theorem \ref{thm_tight_binding_approx} is given as
\begin{equation} \label{eq_tight_bind_proof_14}
\mathfrak{e}[a]:=\lambda^{cont} u^{(1,0)}_{a} +\delta u^{(0,1)}_{a} +r_{a} =\delta(\lambda^{disc} u^{(1,0)}_{a}+u^{(0,1)}_{a})+r_{a}.
\end{equation}
This concludes the proof of the second part of Theorem \ref{thm_tight_binding_approx}.

\subsection{Proof of Proposition \ref{prop_1st_aux_func}} \label{sec_4_1}
We prove Proposition \ref{prop_1st_aux_func} based on a standard iteration argument, similar to the one in \cite[Section 2]{zou2018finite_element}. To facilitate the iteration, we need another auxiliary form defined as follows:
\begin{equation} \label{eq_4th_aux_form}
\mathfrak{a}^{aux,IV}_{\phi,\lambda,\delta}(u,v):= \mathfrak{a}^{aux,II}(u,v)-\lambda ( u, v)_{\mathcal{D}}-\delta \langle \mathcal{T}^{\lambda}[u\big|_{\partial \mathcal{D}}], v\rangle-\lambda \sigma ( |\phi|^2 u, v)_{\mathcal{D}}
\end{equation}
with $\phi\in H^1(\mathcal{D})$ and $$\mathfrak{a}^{aux,II}(u,v)= (\nabla u,\nabla v)_{\mathcal{D}}+\sum_{\substack{1\leq i\leq d \\ \bm{n}\in\mathbb{Z}^2}}\int_{\mathcal{D}_{\bm{n}}^{(i)}}u\int_{\mathcal{D}_{\bm{n}}^{(i)}}\overline{v}$$ (as introduced in \eqref{eq_2nd_aux_form}). For any fixed $\phi$, $\mathfrak{a}^{aux,IV}_{\phi,\lambda,\delta}$ is a sesquilinear form in $u$ and $v$. In particular, it is relatively bounded with respect to the well-posed form $\mathfrak{a}^{aux,II}(u,v)$ (since $\lambda,\delta$ are small), which leads to its invertibility. The following result holds. 
\begin{lemma} \label{lem_4th_form_invert}
There exist $\delta_{1,1},\lambda_{1,1}, M_{1,1}>0$ 
such that for $\delta<\delta_{1,1}$, $\lambda<\lambda_{1,1}$ and $\|\phi\|_{1,\mathcal{D}}<M_{1,1}$, there is a unique solution $u=u_{f}\in H^1(\mathcal{D})$ to the following problem for any $f\in (H^1(\mathcal{D}))^{*}$:
\begin{equation*}
    \mathfrak{a}^{aux,IV}_{\phi,\lambda,\delta}(u,v)=(f,v)_{\mathcal{D}}.
\end{equation*}
Moreover, $u_{f}$ satisfies the estimate
\begin{equation} \label{eq_4th_aux_sol_estimate}
\|u_f\|_{1,\mathcal{D}}^2\leq C\|f\|_{-1,\mathcal{D}}^2\Big(1+\lambda\|\phi\|_{1,\mathcal{D}}^2 \Big),
\end{equation}
where $C>0$ is independent of $\delta,\lambda,\|\phi\|, \delta_{1,1},\lambda_{1,1},$ and $M_{1,1}$.
\end{lemma}
We postpone the proof of this lemma to the end of this subsection. Now we apply an iteration argument to construct a solution to $\mathfrak{a}^{aux,I}_{\lambda,\delta}(u,v)=(f,v)_{\mathcal{D}}$ with the help of Lemma \ref{lem_4th_form_invert}; this will prove the existence of a solution claimed in Proposition \ref{prop_1st_aux_func}. 

{\color{blue}Step 1.} We start with an arbitrary $u^{[0]}\in H^1(\mathcal{D})$ with $\|u^{[0]}\|_{1,\mathcal{D}}\leq M_{1,1}$ and fix $f\in (H^1(\mathcal{D}))^{*}$ (the restriction on the norm of $f$ will be specified in due course). By Lemma \ref{lem_4th_form_invert}, for $\delta<\delta_{1,1}$, $\lambda<\lambda_{1,1}$, there exists a unique $u\in H^1(\mathcal{D})$ that solves  $\mathfrak{a}^{aux,IV}_{u^{[0]},\lambda,\delta}(u,v)=(f,v)_{\mathcal{D}}$. We denote this solution by $u^{[1]}$. Indeed, continuing this procedure, we can inductively construct a sequence $u^{[k]}$ if the following constraint on its norm is satisfied for any $k\geq 0$:
\begin{equation} \label{eq_interation_norm_constraint}
\|u^{[k]}\|_{1,\mathcal{D}}<M_{1,1}.
\end{equation}
Note that \eqref{eq_interation_norm_constraint} is satisfied for $k=0$ due to our initial choice of $u^{(0)}$. Supposing that \eqref{eq_interation_norm_constraint} holds for $0,1,\ldots,k-1$, we apply the Gronwall inequality to \eqref{eq_4th_aux_sol_estimate} and see
\begin{equation} \label{eq_iteration_seq_H1_estimate}
\begin{aligned}
\|u^{[k]}\|_{1,\mathcal{D}}^2\leq C\|f\|_{-1,\mathcal{D}}^2\Big(1+\lambda\|u^{[k-1]}\|_{1,\mathcal{D}}^2 \Big)
\Longrightarrow \|u^{[k]}\|_{1,\mathcal{D}}^2\leq C\|f\|_{-1,\mathcal{D}}^2 
e^{C\lambda \|f\|_{-1,\mathcal{D}}^2} 
\end{aligned}
\end{equation}
with $C$ being independent of $\delta,\lambda,$ and $\|f\|_{-1,\mathcal{D}}$. Hence, we only need to choose a sufficiently small $M_{1,2}>0$ such that $$CM_{1,2}^2 e^{C\lambda_{1,1} M_{1,2}^2} < M_{1,1}^2$$ and restrict $\|f\|_{-1,\mathcal{D}}\leq M_{1,2}$; then the iteration continues inductively, producing a sequence $u^{[k]}$ with the relation
\begin{equation} \label{eq_uk_inductive_relation}
\mathfrak{a}^{aux,IV}_{u^{[k]},\lambda,\delta}(u^{[k+1]},v)=(f,v)_{\mathcal{D}},
\end{equation}
where each $u^{[k]}$ satisfies estimate \eqref{eq_interation_norm_constraint}. Using the Sobolev embedding, the $H^1$ estimate \eqref{eq_interation_norm_constraint} also yields the following $L^p$ inequality for any $p\geq 2$: 
\begin{equation} \label{eq_uk_Lp_norm}
\|u^{[k]}\|_{L^{p}(\mathcal{D})}<C_{p}M_{1,1}
\end{equation}
with $C_p>0$ depending only on $p$.

{\color{blue}Step 2.} Now we prove the convergence of the sequence $u^{[k]}$ by showing that it is a Cauchy sequence. Let $w^{[k]}:=u^{[k+1]}-u^{[k]}$. One can check that $w^{[k]}$ solves the following problem:
\begin{equation} \label{eq_1st_aux_func_proof_1}
\mathfrak{a}^{aux,IV}_{u^{[k]},\lambda,\delta}(w^{[k]},v)=(g,v)_{\mathcal{D}},\quad \text{with } g=\lambda\sigma \big(|u^{[k]}|^2-|u^{[k-1]}|^2 \big)u^{[k]} .
\end{equation}
By the Schwarz inequality,
\begin{equation} \label{eq_1st_aux_func_proof_2}
\begin{aligned}
\Big\| \lambda\sigma \big(|u^{[k]}|^2-|u^{[k-1]}|^2 \big)u^{[k]} \Big\|_{0,\mathcal{D}}
&\leq \Big\| \lambda\sigma \big(|u^{[k]}|+|u^{[k-1]}| \big)|w^{[k-1]}|u^{[k]} \Big\|_{0,\mathcal{D}} \\
&\leq C\lambda \big(\|u^{[k]}\|_{L^4(\mathcal{D})}+\|u^{[k-1]}\|_{L^4(\mathcal{D})}\big)\|u^{[k]}\|_{L^4(\mathcal{D})} \|w^{[k-1]}\|_{0,\mathcal{D}} \\
&\leq C\lambda M_{1,1}^{2} \|w^{[k-1]}\|_{0,\mathcal{D}},
\end{aligned}
\end{equation}
where $C>0$ is independent of $\lambda$ and $k$. Here, \eqref{eq_uk_Lp_norm} is applied to derive the last inequality. Using \eqref{eq_uk_Lp_norm} again, it follows that
\begin{equation} \label{eq_1st_aux_func_proof_3}
\begin{aligned}
\Big\| \lambda\sigma \big(|u^{[k]}|^2-|u^{[k-1]}|^2 \big)u^{[k]} \Big\|_{0,\mathcal{D}}
\leq C\lambda M_{1,1}^{2} \|u^{[k]}-u^{[k-1]}\|_{0,\mathcal{D}}
\leq C\lambda M_{1,1}^{3}.
\end{aligned}
\end{equation}
We choose $\lambda_{1,2}>0$ such that $C\lambda_{1,2} M_{1,1}^{2}<1$ in the above inequality. Then for $\lambda<\min\{\lambda_{1,1},\lambda_{1,2}\}$, we see that \eqref{eq_4th_aux_sol_estimate} can be applied to estimate the solution of \eqref{eq_1st_aux_func_proof_1}, implying that $w^{[k]}$ is squeezing. In fact, we have
\begin{equation*}
\begin{aligned}
\|w^{[k-1]}\|_{1,\mathcal{D}}^2
&\leq C\Big\| \lambda\sigma \big(|u^{[k]}|^2-|u^{[k-1]}|^2 \big)u^{[k]} \Big\|_{0,\mathcal{D}}^2\Big(1+\lambda\|u^{[k-1]}\|_{1,\mathcal{D}}^2 \Big) \\
&\overset{\eqref{eq_interation_norm_constraint},\eqref{eq_1st_aux_func_proof_2}}{\leq} C\lambda^2M_{1,1}^4(1+\lambda_{1,1}M_{1,1}) \|w^{[k-1]}\|_{1,\mathcal{D}}^2.
\end{aligned}
\end{equation*}
Hence, if we shall further select $\lambda_{1,3}$ such that
$$
C\lambda_{1,3}^2M_{1,1}^4(1+\lambda_{1,1}M_{1,1})<\frac{1}{4} ,
$$
then we see that for $\lambda<\min\{\lambda_{1,1},\lambda_{1,2},\lambda_{1,3}\}$, the sequence $u^{[k]}$ constructed in Step 1 is indeed a Cauchy sequence in $H^1(\mathcal{D})$. Denoting its limit by $u_{f}$, one can check that $\mathfrak{a}^{aux,I}_{\lambda,\delta}(u_f,v)=(f,v)_{\mathcal{D}}$ holds for all $v\in H^1(\mathcal{D})$ by pushing $k\to\infty$ in equality \eqref{eq_uk_inductive_relation}. We show only the details for the convergence of the nonlinear part of \eqref{eq_uk_inductive_relation}, \textit{i.e.},  $(|u^{[k]}|^2u^{[k+1]},v)_{\mathcal{D}}\to (|u_f|^2 u_f,v)_{\mathcal{D}}$, while the convergence of the linear part is straightforward. By the Schwarz inequality,
\begin{equation*}
\begin{aligned}
\Big|\big(|u^{[k]}|^2 u^{[k+1]}-|u_{*}|^2 u_{f},v \big)_{\mathcal{D}} \Big|
&=\Big|\big((|u^{[k]}|-|u_*|)(|u^{[k]}|+|u_f|)u^{[k+1]},v \big)_{\mathcal{D}}
+\big(|u_f|^2 (u^{[k+1]}-u_{f}),v \big)_{\mathcal{D}}\Big| \\
&\leq \|u^{[k]}-u_{f}\|_{L^4(\mathcal{D})}\big(\|u^{[k]}\|_{L^4(\mathcal{D})}+\|u_{f}\|_{L^4(\mathcal{D})}\big)\|u^{[k+1]}\|_{L^4(\mathcal{D})}\|v\|_{L^4(\mathcal{D})} \\
&\quad +\|u^{[k+1]}-u_{f}\|_{L^4(\mathcal{D})} \|u_{f}\|_{L^4(\mathcal{D})}^2 \|v\|_{L^4(\mathcal{D})}.
\end{aligned}
\end{equation*}
The right side converges to zero as $k\to\infty$ by the $H^1$ convergence of $u^{[k]}$ and the Sobolev embedding $H^{1}(\mathcal{D})\subset L^{4}(\mathcal{D})$. Hence, we conclude that there exists $\delta_{1,1}$ (introduced in Lemma \ref{lem_4th_form_invert}), $\lambda_{1,4}:=\min\{\lambda_{1,1},\lambda_{1,2},\lambda_{1,3}\}$ and $M_{1,3}:=\min\{M_{1,1},M_{1,2}\}>0$ such that for $\delta<\delta_{1,1}$, $\lambda<\lambda_{1,4}$ and $\|f\|_{-1,\mathcal{D}}\leq M_{1,3}$, there is a solution $u_{f}\in H^1(\mathcal{D})$ to $\mathfrak{a}^{aux,I}_{\lambda,\delta}(u_{f},v)=(f,v)_{\mathcal{D}}$. Moreover, since $\lim_{k} \|u^{[k]}-u_f\|_{1,\mathcal{D}}=0$ and each $u^{[k]}$ satisfies estimate \eqref{eq_iteration_seq_H1_estimate}, we immediately obtain the following $H^1$ estimate:
\begin{equation} \label{eq_1st_aux_sol_estimate_temp}
\|u_f\|_{1,\mathcal{D}}\leq C\|f\|_{-1,\mathcal{D}} \quad \text{with $C$ depending only on $\delta_{1,1},\lambda_{1,4},$ and $M_{1,3}$.}
\end{equation}

{\color{blue}Step 3.} We now prove the uniqueness of a solution to $\mathfrak{a}^{aux,I}_{\lambda,\delta}(u,v)=(f,v)_{\mathcal{D}}$ under appropriate control of the norms. Suppose that $u,\tilde{u}$ both solve this equation and satisfy estimate \eqref{eq_1st_aux_sol_estimate_temp}. Letting $w:=u-\tilde{u}$, we can check that $w$ satisfies:
\begin{equation} \label{eq_1st_aux_func_proof_4}
\mathfrak{a}^{aux,IV}_{u,\lambda,\delta}(w,v)=(g,v)_{\mathcal{D}} \quad \text{with } g=\lambda\sigma \big(|u|^2-|\tilde{u}|^2 \big)\tilde{u} .
\end{equation}
Similarly to \eqref{eq_1st_aux_func_proof_2}, we can prove that
\begin{equation} \label{eq_1st_aux_func_proof_5}
\begin{aligned}
\Big\| \lambda\sigma \big(|u|^2-|\tilde{u}|^2 \big)\tilde{u} \Big\|_{0,\mathcal{D}}
&\leq C\lambda \big(\|u\|_{L^4(\mathcal{D})}+\|\tilde{u}\|_{L^4(\mathcal{D})}\big)\|\tilde{u}\|_{L^4(\mathcal{D})} \|w\|_{0,\mathcal{D}} \\
&\leq C\lambda \|f\|_{-1,\mathcal{D}}^2 \|w\|_{0,\mathcal{D}}
\leq
C\lambda M_{1,3}^2 \|w\|_{0,\mathcal{D}},
\end{aligned}
\end{equation}
where the second inequality follows from estimate \eqref{eq_1st_aux_sol_estimate_temp} and the last one follows from the assumption $\|f\|_{-1,\mathcal{D}}\leq M_{1,3}$. Applying \eqref{eq_1st_aux_sol_estimate_temp} again, we see
$$
\|w\|_{0,\mathcal{D}}\leq \|u\|_{0,\mathcal{D}}+\|\tilde{u}\|_{0,\mathcal{D}}\leq C \|f\|_{0,\mathcal{D}}.
$$
Hence the condition $\|f\|_{0,\mathcal{D}}\leq M_{1,3}$ yields
\begin{equation*}
\Big\| \lambda\sigma \big(|u|^2-|\tilde{u}|^2 \big)\tilde{u} \Big\|_{0,\mathcal{D}}\leq C\lambda M_{1,3}^3 ,
\end{equation*}
where $C$ depends only on $\delta_{1,1},\lambda_{1,4},$ and $M_{1,3}$. We choose $\lambda_{1,5}>0$ such that $ C\lambda_{1,5} M_{1,3}^3 <M_{1,1}$ in the above inequality. Then for $\lambda<\min\{\lambda_{1,4},\lambda_{1,5}\}$, estimate \eqref{eq_4th_aux_sol_estimate} applies to the solution of \eqref{eq_1st_aux_func_proof_4} and leads to
\begin{equation*}
\begin{aligned}
\|w\|_{1,\mathcal{D}}^2
&\leq C\Big\| \lambda\sigma \big(|u|^2-|\tilde{u}|^2 \big)\tilde{u} \Big\|_{0,\mathcal{D}}^2\Big(1+\lambda\|u\|_{1,\mathcal{D}}^2 \Big) \\
&\leq C\lambda^2 M_{1,3}^4 \big(1+\lambda M_{1,3}^2 \big) \|w\|^2_{0,\mathcal{D}}.
\end{aligned}
\end{equation*}
We shall further select $\lambda_{1,6}>0$ such that in the above inequality
$$
C\lambda^2_{1,6} M_{1,3}^4 \big(1+\lambda_{1,6} M_{1,3}^2  \big) \leq \frac{1}{4}.
$$
Then it yields $\|w\|_{1,\mathcal{D}} \leq \frac{1}{2}\|w\|_{0,\mathcal{D}}$, which necessarily requires $w=0$ and hence, $u=\tilde{u}$. This concludes the proof of uniqueness.

In conclusion, the proof of Proposition \ref{prop_1st_aux_func} is completed by recording the appropriate choices of parameters as $\delta_{1}:=\delta_{1,1}$, $\lambda_{1}:=\min\{\lambda_{1,4},\lambda_{1,5},\lambda_{1,6}\}$ and $M_{1}:=M_{1,3}$.

\begin{proof}[Proof of Lemma \ref{lem_4th_form_invert}]
{\color{blue}Step 1.} We first prove that the principal part of $\mathfrak{a}^{aux,IV}_{\phi,\lambda,\delta}(u,v)$, \textit{i.e.}, $\mathfrak{a}^{aux,II}(u,v)$, is continuous and coercive. Indeed, by the Schwarz inequality, it follows that
\begin{equation} \label{eq_4th_form_invert_proof_1}
\mathfrak{a}^{aux,II}(u,u)=(\nabla u,\nabla u)_{\mathcal{D}}+\sum_{\substack{1\leq i\leq d \\ \bm{n}\in\mathbb{Z}^2}}\Big|\int_{\mathcal{D}_{\bm{n}}^{(i)}}u\Big|^2 \leq \|\nabla u\|_{0,\mathcal{D}}^2+\sum_{\substack{1\leq i\leq d \\ \bm{n}\in\mathbb{Z}^2}}\int_{\mathcal{D}_{\bm{n}}^{(i)}}|u|^2=\|u\|_{1,\mathcal{D}}^2.
\end{equation}
On the other hand, by the elementary inequality $\frac{1}{3}|a|^2-2|b|^2\leq |a-b|^2$ and the Poincaré inequality, we have
\begin{equation*}
\frac{1}{3}\int_{\mathcal{D}_{\bm{n}}^{(i)}}|u|^2-2\big|\int_{\mathcal{D}_{\bm{n}}^{(i)}}u\big|^2
=\int_{\mathcal{D}_{\bm{n}}^{(i)}}\Big(\frac{1}{3}|u|^2-2\big|\int_{\mathcal{D}_{\bm{n}}^{(i)}}u\big|^2\Big)
\leq \int_{\mathcal{D}_{\bm{n}}^{(i)}}\big|u-\int_{\mathcal{D}_{\bm{n}}^{(i)}}u \big|^2 \leq C\int_{\mathcal{D}_{\bm{n}}^{(i)}}|\nabla u|^2,
\end{equation*}
for some $C>0$. Moving the volume integral $\big|\int_{\mathcal{D}_{\bm{n}}^{(i)}}u\big|^2$ to the right-hand side and summing over the indices $i$ and $\bm{n}$, we have
\begin{equation} \label{eq_4th_form_invert_proof_2}
\|u\|_{0,\mathcal{D}}^2\leq C\Big((\nabla u,\nabla u)_{\mathcal{D}}+\sum_{\substack{1\leq i\leq d \\ \bm{n}\in\mathbb{Z}^2}}\Big|\int_{\mathcal{D}_{\bm{n}}^{(i)}}u\Big|^2 \Big).
\end{equation}
By \eqref{eq_4th_form_invert_proof_1} and \eqref{eq_4th_form_invert_proof_2}, there exists $C_1\geq 1$ such that
\begin{equation} \label{eq_4th_form_invert_proof_3}
C_{1}^{-1}\|u\|_{1,\mathcal{D}}^2\leq (\nabla u,\nabla u)_{\mathcal{D}}+\sum_{\substack{1\leq i\leq d \\ \bm{n}\in\mathbb{Z}^2}}\Big|\int_{\mathcal{D}_{\bm{n}}^{(i)}}u\Big|^2\leq C_{1}\|u\|_{1,\mathcal{D}}^2,
\end{equation}
which justifies our claim.

{\color{blue}Step 2.} Now, we see the reason for which the form $\mathfrak{a}^{aux,IV}_{\phi,\lambda,\delta}$ is well posed. For $\lambda,\delta$ and $\|\phi\|_{1,\mathcal{D}}$ being sufficiently small, the form $\mathfrak{a}^{aux,IV}_{\phi,\lambda,\delta}$ differs from its well-posed principle part $\mathfrak{a}^{aux,II}$ only by a perturbative term. Indeed, the Schwarz inequality and Sobolev embedding give rise to the following estimates:
\begin{equation} \label{eq_4th_form_invert_proof_4}
\big|\lambda ( u, v)_{\mathcal{D}} \big|\leq C_2\lambda \|u\|_{1,\mathcal{D}}\|v\|_{1,\mathcal{D}},
\end{equation}
\begin{equation} \label{eq_4th_form_invert_proof_5}
\big|\delta \langle \mathcal{T}^{\lambda}[u\big|_{\partial \mathcal{D}}], v\rangle\big| \leq C_{3} |\delta| \|\mathcal{T}^{\lambda}\|_{\mathcal{B}(H^{1/2}(\partial \mathcal{D}),H^{-1/2}(\partial \mathcal{D}))}  \|u\|_{1,\mathcal{D}}\|v\|_{1,\mathcal{D}}, 
\end{equation}
\begin{equation} \label{eq_4th_form_invert_proof_6}
\big|\lambda \sigma ( |\phi|^2 u, v)_{\mathcal{D}}\big|  \leq C_{4}\lambda\|\phi\|^{2}_{1,\mathcal{D}} \|u\|_{1,\mathcal{D}}\|v\|_{1,\mathcal{D}},
\end{equation}
where all the constants $C_k>0$ are independent of the parameters $\lambda,\delta,$ and $\|\phi\|$. By \eqref{eq_4th_form_invert_proof_3}-\eqref{eq_4th_form_invert_proof_6}, it is sufficient to choose $\delta_{1,1},\lambda_{1,1}, M_{1,1}>0$ such that
\begin{equation} \label{eq_4th_form_invert_proof_7}
\left\{
\begin{aligned}
&C_2\lambda_{1,1} \leq C_1^{-1}/6, \\
&C_3\delta_{1,1} \sup_{|\lambda|<\lambda_0}\|\mathcal{T}^{\lambda}\|_{\mathcal{B}(H^{1/2}(\partial \mathcal{D}),H^{-1/2}(\partial \mathcal{D}))} \leq C_1^{-1}/6, \\
&C_4\lambda_{1,1} M_{1,1}^2 \leq C_1^{-1}/6 ,
\end{aligned}
\right.
\end{equation}
where $\lambda_0$ is chosen in Proposition \ref{prop_d2n_map} which guarantees $\sup_{|\lambda|<\lambda_0}\|\mathcal{T}^{\lambda}\|_{\mathcal{B}(H^{1/2}(\partial \mathcal{D}),H^{-1/2}(\partial \mathcal{D}))}<\infty$. Then we see from \eqref{eq_4th_aux_form} that $\mathfrak{a}^{aux,IV}_{\phi,\lambda,\delta}$ is continuous and coercive when $\delta<\delta_{1,1}$, $\lambda<\lambda_{1,1}$ and $\|\phi\|_{1,\mathcal{D}}<M_{1,1}$ that
\begin{equation} \label{eq_4th_form_invert_proof_8}
\frac{1}{2}C_{1}^{-1}\|u\|_{1,\mathcal{D}}^2\leq \mathfrak{a}^{aux,IV}_{\phi,\lambda,\delta}(u,u)\leq \frac{3}{2}C_{1}\|u\|_{1,\mathcal{D}}^2.
\end{equation}
Estimate \eqref{eq_4th_form_invert_proof_8}, together with the Lax-Milgram theorem, allows us to conclude that there exists a unique solution $u_{f}\in H^1(\mathcal{D})$ to the problem $\mathfrak{a}^{aux,IV}_{\phi,\lambda,\delta}(u,v)=(f,v)_{\mathcal{D}}$ for any $f\in L^2(\mathcal{D})$.

{\color{blue}Step 3.} Finally, we prove estimate \eqref{eq_4th_aux_sol_estimate}. Taking $v=u$ in $\mathfrak{a}^{aux,IV}_{\phi,\lambda,\delta}(u,v)=(f,v)_{\mathcal{D}}$ and applying \eqref{eq_4th_form_invert_proof_3}-\eqref{eq_4th_form_invert_proof_6}, we obtain that
\begin{equation*}
\begin{aligned}
\frac{2}{3}C_{1}^{-1}\|u\|_{1,\mathcal{D}}^2
&\leq \Big|(\nabla u,\nabla v)_{\mathcal{D}}+\sum_{\substack{1\leq i\leq d \\ \bm{n}\in\mathbb{Z}^2}}\int_{\mathcal{D}_{\bm{n}}^{(i)}}u\int_{\mathcal{D}_{\bm{n}}^{(i)}}\overline{v}-\lambda ( u, v)_{\mathcal{D}}-\delta \langle \mathcal{T}^{\lambda}[u\big|_{\partial \mathcal{D}}], v\rangle\Big| \\
&\leq |(f,u)_{\mathcal{D}}|+|\lambda\sigma ( |\phi|^2 u, u)_{\mathcal{D}}| \\
&\leq \|f\|_{-1,\mathcal{D}}\|u\|_{1,\mathcal{D}}+C_{4}\lambda\|\phi\|^{2}_{1,\mathcal{D}} \|u\|_{1,\mathcal{D}}^2.
\end{aligned}
\end{equation*}
Hence, for any $M>0$, it holds that
\begin{equation*}
\frac{2}{3}C_{1}^{-1}\|u\|_{1,\mathcal{D}}^2\leq M^2\|f\|_{-1,\mathcal{D}}^2+\frac{1}{4M^2}\|u\|_{0,\mathcal{D}}^2+C_{4}\lambda\|\phi\|^{2}_{1,\mathcal{D}} \|u\|_{1,\mathcal{D}}^2.
\end{equation*}
Select $M>0$ such that $\frac{1}{4M^2}<\frac{1}{3}C_{1}^{-1}$. Then we see
\begin{equation} \label{eq_4th_form_invert_proof_9}
\begin{aligned}
\big(\frac{1}{3}C_{1}^{-1}-C_{4}\lambda\|\phi\|^{2}_{1,\mathcal{D}} \big)\|u\|_{1,\mathcal{D}}^2 \leq M^2\|f\|_{-1,\mathcal{D}}^2.
\end{aligned}
\end{equation}
We further impose the following constraints on $\lambda_{1,1}$ and $M_{1,1}$:
\begin{equation} \label{eq_4th_form_invert_proof_10}
3C_{1}C_{4}\lambda_{1,1}M_{1,1}^2<\frac{1}{2} .
\end{equation}
Then for $\lambda<\lambda_{1,1}$ and $\|\phi\|_{1,\mathcal{D}}\leq M_{1,1}$, \eqref{eq_4th_form_invert_proof_9} leads to
\begin{equation*}
\|u\|_{1,\mathcal{D}}^2 \leq \frac{3C_1M^2\|f\|_{-1,\mathcal{D}}^2}{1-3C_1C_4\lambda \|\phi\|_{1,\mathcal{D}}^2}
\leq C\|f\|_{-1,\mathcal{D}}^2\big(1+\lambda\|\phi\|_{1,\mathcal{D}}^2 \big)
\end{equation*}
with $C:=6C_1M^2\max\{1,3C_1C_4\}$. This concludes the proof of \eqref{eq_4th_aux_sol_estimate}.

In conclusion, by choosing $\lambda_{1,1},\delta_{1,1},M_{1,1}>0$ that satisfy conditions \eqref{eq_4th_form_invert_proof_7} and \eqref{eq_4th_form_invert_proof_10} (or simply setting them to be sufficiently small), all the statements in Lemma \ref{lem_4th_form_invert} hold.
\end{proof}

\subsection{Proof of Proposition \ref{prop_ansatz_aux_sol}} \label{sec_4_2} \enspace

{\color{blue}Step 1.} The proof of the first part is straightforward. Suppose that $f=\sum_{\substack{1\leq i\leq d \\ \bm{n}\in\mathbb{Z}^2}}a_{\bm{n}}^{(i)}\mathbbm{1}_{\mathcal{D}_{\bm{n}}^{(i)}}$. By the definition of $u^{(0,0)}_{a}$, $u^{(1,0)}_{a}$ and $u^{(0,1)}_{a}$ in \eqref{eq_ansatz_aux_sol_2} and \eqref{eq_2nd_aux_form}, one sees that they solve the following problems, respectively:
\begin{equation*}
\mathfrak{a}^{aux,II}(u^{(0,0)}_{a},v)=(f,v)_{\mathcal{D}},
\end{equation*}
\begin{equation*}
\mathfrak{a}^{aux,II}(u^{(1,0)}_{a},v)=\big(u^{(0,0)}_{a}+\sigma |u^{(0,0)}_{a}|^2u^{(0,0)}_{a},v\big)_{\mathcal{D}},
\end{equation*}
\begin{equation*}
\mathfrak{a}^{aux,II}(u^{(0,1)}_{a},v): = \langle \mathcal{T}^{\lambda}[u^{(0,0)}_{a}\big|_{\partial \mathcal{D}}], v\rangle , 
\end{equation*}
Then, by subtracting the above identities from $\mathfrak{a}^{aux,I}_{\lambda,\delta}(u_{f},v)=(f,v)_{\mathcal{D}}$, one can directly check that the remainder $r_{a}:=u_{f}-u^{(0,0)}_{a}-\lambda u^{(1,0)}_{a}-\delta u^{(0,1)}_{a}$ satisfies equation \eqref{eq_3rd_aux_form}. We also note that identity \eqref{eq_u01_boundary_int} is verified by applying the integration by parts and using the partial differential equation formulation of \eqref{eq_2nd_aux_form}:
\begin{equation*}
\left\{
\begin{aligned}
&-\Delta u^{(0,1)}_{a}+\sum_{\substack{1\leq i\leq d \\ \bm{n}\in\mathbb{Z}^2}}\big(\int_{\mathcal{D}_{\bm{n}}^{(i)}}u^{(0,1)}_{a}\big)\mathbbm{1}_{\mathcal{D}_{\bm{n}}^{(i)}}=0\quad \text{in }\mathcal{D}, \\
&\frac{\partial u^{(0,1)}_{a}}{\partial \nu}=\mathcal{T}^{0}[u^{(0,1)}_{a}\big|_{\mathcal{D}}] \quad \text{on }\partial \mathcal{D}^{-}.
\end{aligned}
\right.
\end{equation*}

{\color{blue}Step 2.} The second part of Proposition \ref{prop_ansatz_aux_sol}, \textit{i.e.}, the solvability of the following problem: \begin{equation} \label{eq:pdef} \mathfrak{a}^{aux,III}_{a,\lambda,\delta}(u,v)=(f,v)_{\mathcal{D}},
\end{equation} is proved following the same lines as in Proposition \ref{prop_1st_aux_func}. Here we only sketch the idea and skip all the details. As in the proof of Proposition \ref{prop_1st_aux_func}, the key is to construct the solution through an iteration argument. For problem \eqref{eq:pdef}, the iteration is based on the following auxiliary form:
\begin{equation} \label{eq_5th_aux_form}
\begin{aligned}
\mathfrak{a}^{aux,V}_{a,\phi,\lambda,\delta}(u,v)
:=& (\nabla u,\nabla v)_{\mathcal{D}}+\sum_{\substack{1\leq i\leq d \\ \bm{n}\in\mathbb{Z}^2}}\int_{\mathcal{D}_{\bm{n}}^{(i)}}u\int_{\mathcal{D}_{\bm{n}}^{(i)}}\overline{v}-\lambda(u,v)_{\mathcal{D}}-\lambda\sigma (|u^{(0,0)}_{a}|^2 u,v)_{\mathcal{D}} \\
&-2\lambda\sigma \left(\Re(\phi\overline{u^{(0,0)}_{a}})u,v\right)_{\mathcal{D}}-2\lambda\sigma \left(\Re((\lambda u^{(1,0)}_{a}+\delta u^{(0,1)}_{a})\overline{u^{(0,0)}_{a}})u,v\right)_{\mathcal{D}} \\
&-\lambda\sigma \left((u^{(0,0)}_{a}+\lambda u^{(1,0)}_{a}+\delta u^{(0,1)}_{a})\overline{u^{(0,0)}_{a}}u,v\right)_{\mathcal{D}} \\ 
&-\lambda\sigma (|\lambda u^{(1,0)}_{a}+\delta u^{(0,1)}_{a}|^2 u,v)_{\mathcal{D}} \\
&-\lambda\sigma\left( (u^{(0,0)}_{a}+\lambda u^{(1,0)}_{a}+\delta u^{(0,1)}_{a}) (\overline{\lambda u^{(1,0)}_{a}+\delta u^{(0,1)}_{a}})u ,v \right)_{\mathcal{D}} \\
&-\lambda\sigma \left( \overline{\phi}(u^{(0,0)}_{a}+\lambda u^{(1,0)}_{a}+\delta u^{(0,1)}_{a})u,v \right)_{\mathcal{D}} -\lambda\sigma (|\phi|^2 u,v)_{\mathcal{D}} \\
&-2\lambda\sigma \left( \phi\Re( (\overline{\lambda u^{(1,0)}_{a}+\delta u^{(0,1)}_{a}})u,v ) \right)_{\mathcal{D}} \\
&-\delta\langle \mathcal{T}^{\lambda}[u\big|_{\partial \mathcal{D}}],v \rangle \\
&-\lambda\sigma \left((u^{(0,0)}_{a}+\lambda u^{(1,0)}_{a}+\delta u^{(0,1)}_{a}) u^{(0,0)}_{a}\overline{\phi},v\right)_{\mathcal{D}} \\
&-\lambda\sigma\left( (u^{(0,0)}_{a}+\lambda u^{(1,0)}_{a}+\delta u^{(0,1)}_{a})(\lambda u^{(1,0)}_{a}+\delta u^{(0,1)}_{a})\overline{\phi},v \right)_{\mathcal{D}}. \\
\end{aligned}
\end{equation}
Although this form is more complex compared to the auxiliary form $\mathfrak{a}^{aux,IV}_{\phi,\lambda,\delta}(u,v)$ in the proof of Proposition \ref{prop_1st_aux_func}, we remark that they share a similar structure. In fact, both the two forms are small perturbations of the well-posed principle part $$\mathfrak{a}^{aux,II}(u,v)=(\nabla u,\nabla v)_{\mathcal{D}}+\sum_{\substack{1\leq i\leq d \\ \bm{n}\in\mathbb{Z}^2}}\int_{\mathcal{D}_{\bm{n}}^{(i)}}u\int_{\mathcal{D}_{\bm{n}}^{(i)}}\overline{v},$$ when $\lambda,\delta,$ and $a$ are appropriately chosen and $\|\phi\|$ is appropriately controlled. As shown in the proof of Proposition \ref{prop_1st_aux_func}, this facilitates the construction of a sequence $u^{[k]}$ with the following relation:
\begin{equation} \label{eq_ansatz_aux_sol_proof_1}
\mathfrak{a}^{aux,V}_{a,u^{[k-1]},\lambda,\delta}(u^{[k]},v)=(f,v)_{\mathcal{D}}\quad (k=1,2,\ldots).
\end{equation}
Recall that this construction requires a $k-$independent bound of $\|u^{[k]}\|_{1,\mathcal{D}}$, similar to \eqref{eq_interation_norm_constraint}. Fortunately, this can be obtained by the Gronwall inequality, as we did in \eqref{eq_iteration_seq_H1_estimate}; in fact, the availability of this uniform estimate is a natural consequence of the polynomial dependence of our auxiliary form $\mathfrak{a}^{aux,V}_{a,\phi,\lambda,\delta}(u,v)$ on $\phi$ and $u$. We note the following estimates, which are derived similarly to \eqref{eq_tight_bind_proof_4}-\eqref{eq_tight_bind_proof_9}:
\begin{equation*}
\Big|\lambda\sigma \left(\Re(\phi\overline{u^{(0,0)}_{a}})u,u\right)_{\mathcal{D}}\Big|
\leq \lambda C \|\phi\|_{1,\mathcal{D}} \|a\|_{\ell^2(\mathbb{Z}^2;\mathbb{C}^{d})}\|u\|_{1,\mathcal{D}}^2,
\end{equation*}
\begin{equation*}
\Big|\left( \overline{\phi}(u^{(0,0)}_{a}+\lambda u^{(1,0)}_{a}+\delta u^{(0,1)}_{a})u,u \right)_{\mathcal{D}}  \Big|
\leq \lambda C\|\phi\|_{1,\mathcal{D}}(1+\|a\|_{\ell^2(\mathbb{Z}^2;\mathbb{C}^{d})}^{3}) \|u\|_{1,\mathcal{D}}^2,
\end{equation*}
\begin{equation*}
\Big|\lambda\sigma (|\phi|^2 u,u)_{\mathcal{D}}\Big|
\leq \lambda C \|\phi\|_{1,\mathcal{D}}^2 \|u\|_{1,\mathcal{D}}^2,
\end{equation*}
\begin{equation*}
\Big| \lambda\sigma \left( \phi\Re( (\overline{\lambda u^{(1,0)}_{a}+\delta u^{(0,1)}_{a}})u,u ) \right)_{\mathcal{D}}\Big|
\leq \lambda(\lambda+\delta)C\|\phi\|_{1,\mathcal{D}}(1+\|a\|_{\ell^2(\mathbb{Z}^2;\mathbb{C}^{d})}^{3}) \|u\|_{1,\mathcal{D}}^2,
\end{equation*}
\begin{equation*}
\Big| \lambda\sigma \left((u^{(0,0)}_{a}+\lambda u^{(1,0)}_{a}+\delta u^{(0,1)}_{a}) u^{(0,0)}_{a}\overline{\phi},u\right)_{\mathcal{D}} \Big| \leq \lambda C \|a\|_{\ell^2(\mathbb{Z}^2;\mathbb{C}^d)}(1+\|a\|_{\ell^2(\mathbb{Z}^2;\mathbb{C}^{d})}^{3}) \|\phi\|_{1,\mathcal{D}} \|u\|_{1,\mathcal{D}},
\end{equation*}
\begin{equation*}
\Big| \lambda\sigma\left( (u^{(0,0)}_{a}+\lambda u^{(1,0)}_{a}+\delta u^{(0,1)}_{a})(\lambda u^{(1,0)}_{a}+\delta u^{(0,1)}_{a})\overline{\phi},u \right)_{\mathcal{D}} \Big| \leq \lambda C (1+\|a\|_{\ell^2(\mathbb{Z}^2;\mathbb{C}^{d})}^{3})^2 \|\phi\|_{1,\mathcal{D}} \|u\|_{1,\mathcal{D}}.
\end{equation*}
Hence, by suitably choosing $\lambda_2,\delta_2,M_2>0$ and when $\lambda<\lambda_2,\delta<\delta_2,\|a\|_{\ell^2(\mathbb{Z}^2;\mathbb{C}^d)}<M_2$, we can see that the right-hand sides of the above inequalities are all controlled by
\begin{equation*}
C\lambda \|\phi\|_{1,\mathcal{D}}\|u\|_{1,\mathcal{D}}(1+\|u\|_{1,\mathcal{D}}+\|\phi\|_{1,\mathcal{D}}\|u\|_{1,\mathcal{D}}),
\end{equation*}
where $C$ depends only on $\lambda_2,\delta_2,$ and $M_2$. On the other hand, all the remaining terms in \eqref{eq_5th_aux_form}, if taken $v=u$, are relatively bounded by the principal term $\|u\|^2_{1,\mathcal{D}}$. In short, by choosing sufficiently small $\lambda_2,\delta_2,$ and $M_2>0$, the following holds when $\lambda<\lambda_2,\delta<\delta_2,\|a\|_{\ell^2(\mathbb{Z}^2;\mathbb{C}^d)}<M_2$:
\begin{equation} \label{eq_ansatz_aux_sol_proof_2}
C^{-1}\|u^{[k]}\|_{1,\mathcal{D}}^2
\leq \|f\|_{-1,\mathcal{D}}^2
+\lambda C\|u^{[k-1]}\|_{1,\mathcal{D}} \|u^{[k]}\|_{1,\mathcal{D}}\big(1+\|u^{[k]}\|_{1,\mathcal{D}}+\|u^{[k-1]}\|_{1,\mathcal{D}} \|u^{[k]}\|_{1,\mathcal{D}} \big),
\end{equation}
where $C$ depends only on $\lambda_2,\delta_2,$ and $M_2$. This estimate, arguing similarly as in Step 1 of the proof of Proposition \ref{prop_1st_aux_func}, leads to the uniform boundedness of $u^{[k]}$. Then, following the lines of Step 2 of Proposition \ref{prop_1st_aux_func}, one can check that the sequence indeed converges in $H^1(\mathcal{D})$, whose limit solves the problem $\mathfrak{a}^{aux,III}_{a,\lambda,\delta}(u,v)=(f,v)_{\mathcal{D}}$.

The uniqueness of the solution to $\mathfrak{a}^{aux,III}_{a,\lambda,\delta}(u,v)=(f,v)_{\mathcal{D}}$ is also proved in a similar way to Proposition \ref{prop_1st_aux_func}, which requires further restriction on the size of $\lambda_2,\delta_2,$ and $M_2$ and the amplitude $\|f\|_{-1,\mathcal{D}}$ of the source; we skip the details here.

\section{Discrete gap solitons: Proof of Theorem \ref{thm_defect_gap_solitons}} 
\label{sec_gap_soliton}

\subsection{Roadmap} \label{subsec_soliton_road}

Motivated by the study of the whole-space gap-solitons in \cite{pankov2010soliton_discrete,Pankov2005soliton_continuous}, our proof for the existence of defect gap-solitons proceeds in two steps. Due to the loss of compactness of the unit ball in $\ell^2(\mathbb{Z}^2;\mathbb{C}^{d})$, we first find a sequence of (approximate) solutions $a^{[k]}$ to \eqref{eq_discrete_station_eq_defect} in finite-dimensional spaces ($k-$periodic spaces as introduced later). Then we extract a subsequence $a^{[k_i]}$ which (weakly) converges to a solution to \eqref{eq_discrete_station_eq_defect} in $\ell^2(\mathbb{Z}^2;\mathbb{C}^{d})$.

To start, we define the following functional space:
\begin{equation*}
    \ell^2_{k}(\mathbb{Z}^2;\mathbb{C}^{d}):=\big\{a\in \ell^{2}_{loc}(\mathbb{Z}^2;\mathbb{C}^{d}):\, a_{\bm{n}+k\bm{e}_i}=a_{\bm{n}},\, i=1,2\big\}.
\end{equation*}
Note that $\ell^2_{k}(\mathbb{Z}^2;\mathbb{C}^{d})\simeq \mathbb{C}^{k^2d}$ is a finite-dimensional Hilbert space equipped with the local $\ell^2-$inner product and norm:
\begin{equation*}
(a,b)_{\ell^2_{k}(\mathbb{Z}^2;\mathbb{C}^{d})}:=(a\cdot\chi_{Y_k},b\cdot\chi_{Y_k})_{\ell^2(\mathbb{Z}^2;\mathbb{C}^{d})},\quad \|a\|_{\ell^2_{k}(\mathbb{Z}^2;\mathbb{C}^{d})}:=\sqrt{(a,a)_{\ell^2_{k}(\mathbb{Z}^2;\mathbb{C}^{d})}}.
\end{equation*}
Here, $Y_k:=\{\bm{n}:0\leq n_1,n_2<k\}\subset \mathbb{Z}^2$ denotes the primitive $k-$cell. The $\ell^{p}_{k}$ norm for any $p\geq 1$ is defined similarly; note that all $\ell^{p}_{k}$ norms on $\ell^2_{k}(\mathbb{Z}^2;\mathbb{C}^{d})$ are equivalent. In the sequel, we denote $\ell^{p}=\ell^{p}(\mathbb{Z}^2;\mathbb{C}^d)$ and $\ell^{p}_{k}=\ell^{p}_{k}(\mathbb{Z}^2;\mathbb{C}^d)$ for brevity. The following operators between $\ell^{p}$ and $\ell^{p}_{k}$ are useful: the truncation operator $R^{[k]}$ is defined to restrict a sequence on $Y_k$:
\begin{equation*}
(R^{[k]} a)_{\bm{n}}:=
\left\{
\begin{aligned}
& a_{\bm{n}},\quad \bm{n}\in Y_k, \\
& 0,\quad \text{else,}
\end{aligned}
\right.
\end{equation*}
and the periodization operator $S_{k}$ first truncates a sequence onto $Y_k$, then extends it periodically:
\begin{equation*}
(S^{[k]} a)_{\bm{n}}:=(R^{[k]} a)_{(n_1 \text{ mod }k,n_2 \text{ mod } k)}.
\end{equation*}
Due to the periodicity of $\mathcal{C}$ as shown in Proposition \ref{prop_cap_operator}, it is invariant on $\ell^p_{k}$; we denote the restriction of $\mathcal{C}$ on $\ell^2_{k}$ by $\mathcal{C}^{[k]}$. The energy functional on $\ell^2_{k}$ associated with the nonlinear problem \eqref{eq_discrete_station_eq} is given by
\begin{equation*}
J^{[k]}(a):=\frac{1}{2}(\mathcal{C}^{[k]}a,a)_{\ell^2_{k}}-\frac{1}{2}\lambda\|a\|_{\ell^2_{k}}^2+\frac{1}{2}(S^{[k]}VR^{[k]}a,a)_{\ell^2_k}-\frac{1}{4}\lambda\sigma \|a\|_{\ell^4_{k}}^4.
\end{equation*}
By definition, it is directly seen that $J^{[k]}$ is a real-valued $C^1$ functional and its derivative is given by
\begin{equation*}
(J^{[k]})^{\prime}(a)=\mathcal{C}^{[k]}a-\lambda a+S^{[k]}VR^{[k]}a-\lambda\sigma |a|^2a .
\end{equation*}

The main result of the first step is that for any $\lambda\in \mathcal{I}$ and $k$ being large, if the size of the defect $V$ is suitably controlled, then $J^{[k]}$ has a critical point with a uniformly bounded norm. Moreover, the projections of the critical points onto a specific vector is uniformly bounded from below. Specifically, we have:
\begin{proposition} \label{prop_periodic_gap_solitons}
Suppose that $\mathcal{I}$ is a spectral gap of $\mathcal{C}$ as stated in Theorem \ref{thm_gap_solitons} and fix $\lambda\in \mathcal{I}$. There exists $z_0\in \ell^2$, $k_0\in\mathbb{N}$, $M_0>0$ such that for any $k>k_0$ and
\begin{equation*}
2\|V\|_{\mathcal{B}(\ell^\infty,\ell^{1})}<\delta:=\text{dist}\{\lambda,\text{Spec}(\mathcal{C})\},
\end{equation*}
there exists $a^{[k]}\in \ell^2_{k}$ that is a critical point of the functional $J^{[k]}$ and satisfies
\begin{equation} \label{eq_periodic_gap_solitons_1}
\|a^{[k]}\|_{\ell^2_{k}}\leq M_1,\quad (R^{[k]}a^{[k]},z_0)_{\ell^2}\geq m_1,
\end{equation}
where $m_1,M_1>0$ are independent of $k$.
\end{proposition}
The proof is given in Section \ref{subsec_soliton_periodic}. The key idea is that the lower-gap and upper-gap spectrum of $\mathcal{C}^{[k]}$ is linked by the nonlinear functional $J^{[k]}$ and forms a valley-top structure, which leads to the existence of critical points of $J^{[k]}$.

Due to the uniform boundedness of $\|a^{[k]}\|_{\ell^2_{k}}$ as shown in \eqref{eq_periodic_gap_solitons_1}, it is possible to extract a subsequence from $a^{[k]}$ that converges weakly to a bounded state $a\in \ell^2$, which solves \eqref{eq_discrete_station_eq} and decays exponentially as we will prove in Section \ref{subsec_soliton_concentration}.
\begin{proposition} \label{prop_concentration}
There exists a subsequence $a^{[k_{i}]}$ of $a^{[k]}$ which converges weakly in $\ell^{2}_{loc}(\mathbb{Z}^2;\mathbb{C}^{d})$ to a nonzero limit $a\in \ell^{2}$. Moreover, $a$ solves \eqref{eq_discrete_station_eq_defect} and decays exponentially at infinity.
\end{proposition}
This concludes the proof of Theorem \ref{thm_gap_solitons} on the existence of defect gap solitons. We point out that the second condition in \eqref{eq_periodic_gap_solitons_1} ensures that the limiting bounded state $a\in \ell^2$ is nontrivial because $(a,z_0)_{\ell^2}=\lim_{i\to\infty}(R^{[k_i]}a^{[k_i]},z_0)_{\ell^2}>0$ as we will prove in Section \ref{subsec_soliton_concentration}. Compared with the study of periodic gap solitons in \cite{pankov2010soliton_discrete,Pankov2005soliton_continuous,pelinovsky2011localization}, where the nontriviality of the limit is guaranteed by the periodicity of the structure and a concentration argument, our problem \eqref{eq_discrete_station_eq_defect} loses periodicity due to the presence of the defect $V$. It requires a stronger restriction on the finite-dimensional solutions constructed in Proposition \ref{prop_periodic_gap_solitons} such that the limit we find in Proposition \ref{prop_concentration} is nonzero. This motivates the designation of \eqref{eq_periodic_gap_solitons_1}.

\subsection{Periodic in-gap solution: Proof of Proposition \ref{prop_periodic_gap_solitons}} \label{subsec_soliton_periodic}

\subsubsection{Preliminaries}

We first present some preliminaries. The first lemma is a standard statement on the weak compactness of uniformly bounded periodic sequences. Note that we denote all the $
\ell^p-\ell^q$ dual products ($1/p+1/q=1$) by $(\cdot,\cdot)_{\ell^2}$ for brevity. 
\begin{lemma} \label{lem_local_bound_weak_converge}
Let $a^{[k]}\in \ell_k^2$ ($k=1,2,\cdots$) satisfy $\|a^{[k]}\|_{\ell_k^2}\leq C$ for some $k-$independent constant $C>0$. Then $a^{[k]}$ admits a subsequence $a^{[k_i]}$ that converges $\ell^\infty-$weakly to $a\in \ell^2$ in the sense that
\begin{equation} \label{eq_local_bound_weak_converge_1}
(a^{[k_i]},b)_{\ell^2} \to (a,b)_{\ell^2}, \quad \text{for any $b\in \ell^1$} ,
\end{equation}
with $\|a\|_{\ell^2}\leq C$. Moreover, $R^{[k_i]}a^{[k_i]}$ converges $\ell^2-$weakly to $a$:
\begin{equation} \label{eq_local_bound_weak_converge_2}
(R^{[k_i]}a^{[k_i]},b)_{\ell^2} \to (a,b)_{\ell^2}, \quad \text{for any $b\in \ell^2$} .
\end{equation}
\end{lemma}

\begin{proof}

Without loss of generality, we assume that the sequence $a^{[k]}$ is real-valued and non-negative. To obtain \eqref{eq_local_bound_weak_converge_1}, we first show that \eqref{eq_local_bound_weak_converge_1} holds for compactly supported test functions, then we apply a standard truncation argument to complete the proof.

{\color{blue}Step 1.} Let $c^{[i,k]}:=R^{[i]}a^{[k]}\in \ell^2$. First, for any fixed $i\geq 1$, $c^{[i,k]}$ is supported in the finite-dimensional space $Y_i$. Hence, by the boundedness of $a^{[k]}$, we can extract a subsequence of $c^{[i,k]}$, denoted as $\{c^{[i,k_j]}\}_{j\geq 1}$, which converges in $\ell^2$, whose limit is apparently supported in $Y_i$. Next, we can select a subsequence of $\{k_j\}_{j\geq 1}$, which is still denoted as $\{k_j\}$, such that $\{c^{[i+1,k_j]}\}_{j\geq 1}$ strongly converges to an element supported in $Y_{i+1}\supset Y_{i}$. Inductively, one can construct a sequence by a diagonal argument, denoted as $\{c^{[i,k_i]}\}_{i\geq 1}$, such that for any $j\geq 1$
\begin{equation*}
R^{[j]}c^{[i,k_i]} \text{ converges to a compactly supported function } d^{[j]} \text{ as $i\to\infty$.}
\end{equation*}
Moreover, we note that the limit $a:=\lim_{j\to\infty}d^{[j]}$ exists because $d^{[j]}$ is monotonic ($d^{[j]}$ agrees with $d^{[j-1]}$ on $Y_{j-1}$ and is non-negative outside $Y_{j-1}$), and $a\in \ell^2$ thanks to the uniform bound of $a^{[k]}$:
\begin{equation*}
\|d^{[j]}\|_{\ell^2}=\lim_{i\to \infty} \|R^{[j]}c^{[i,k_i]}\|_{\ell^2}
=\lim_{i\to \infty} \|a^{[k_i]}\|_{\ell^2_{j}}\leq \sup_{k}\|a^{[k]}\|_{\ell^2_{k}}\leq C \,\, (\forall j\geq 1),
\end{equation*}
which implies that $\|a\|_{\ell^2}=\lim_{j\to\infty}\|d^{[j]}\|_{\ell^2}\leq C$. Now, the $\ell^{\infty}-$weak convergence of $a^{[k_i]}$ follows immediately. For any compactly supported $b\in \ell^2$, $(a^{[k_i]},b)_{\ell^2}=(R^{[i]}a^{[k_i]},b)_{\ell^2}$ for $i$ being sufficient large and hence,
\begin{equation*}
\lim_{i\to \infty}(a^{[k_i]},b)_{\ell^2}=\lim_{i\to \infty}(R^{[i]}a^{[k_i]},b)_{\ell^2}
=(a,b)_{\ell^2}.
\end{equation*}

{\color{blue}Step 2.} For $b\in \ell^1$, we decompose it into a compactly supported part and a remainder. Namely, for any $\epsilon>0$, let $b_0$ be compactly supported and satisfy $\|b-b_0\|_{\ell^1}\leq \epsilon$. Then we see
\begin{equation*}
\begin{aligned}
|(a^{[k_i]}-a,b)_{\ell^2}|
&\leq |(a^{[k_i]}-a,b_0)_{\ell^2}|+|(a^{[k_i]}-a,b-b_0)_{\ell^2}| \\
&\leq |(a^{[k_i]}-a,b_0)_{\ell^2}|
+\|a^{[k_i]}-a\|_{\ell^\infty}\|b-b_0\|_{\ell^1} \\
&\leq |(a^{[k_i]}-a,b_0)_{\ell^2}| +2C\epsilon .
\end{aligned}
\end{equation*}
Passing $i\to \infty$ and by the arbitrariness of $\epsilon$, one sees $(a^{[k_i]}-a,b)_{\ell^2}\to 0$ and concludes the proof of \eqref{eq_local_bound_weak_converge_1}.

{\color{blue}Step 3.} Now we prove \eqref{eq_local_bound_weak_converge_2}. As in Step 2, we first decompose $b\in \ell^2$ into a compactly supported part $b_0$ and a remainder $b_1$ with $\|b_1\|_{\ell^2}<\epsilon$. The key point is that the part of $R^{[k_i]}a^{[k_i]}-a$ outside $\text{supp}(b_0)$ is uniformly bounded, \textit{i.e.},
\begin{equation*}
\big\|\big(R^{[k_i]}a^{[k_i]}-a\big)\cdot (1-\chi_{\text{supp}(b_0)})\big\|_{\ell^2}\leq \|R^{[k_i]}a^{[k_i]}-a\|_{\ell^2}\leq \|a^{[k_i]}\|_{\ell^2_k}+\|a\|_{\ell^2}\leq 2C . 
\end{equation*}
Hence, for $i$ being large such that $\text{supp}(b_0)\subset \text{supp}(R^{[k_i]}a^{[k_i]})$, one sees
\begin{equation*}
\begin{aligned}
|(R^{[k_i]}a^{[k_i]}-a,b)_{\ell^2}|
&\leq |(R^{[k_i]}a^{[k_i]}-a,b_0)_{\ell^2}|+\big|\big(\big(R^{[k_i]}a^{[k_i]}-a\big)\cdot (1-\chi_{\text{supp}(b_0)}),b_1\big)_{\ell^2}\big| \\
&\leq |(a^{[k_i]}-a,b_0)_{\ell^2}|
+2C\epsilon .
\end{aligned}
\end{equation*}
This concludes the proof of \eqref{eq_local_bound_weak_converge_2} by passing $i\to \infty$ and using the arbitrariness of $\epsilon$.
\end{proof}

The next lemma regards the $\ell^{p}_{k}$ boundedness of the spectral projection associated with the capacitance operator $C^{[k]}$. To be precise, we first fix some notation. Let $$\mathcal{I}_{+}:=\text{Spec}(\mathcal{C})\cap (\sup \mathcal{I},\infty) \quad \text{and}  \quad \mathcal{I}_{-}:=\text{Spec}(\mathcal{C})\cap (-\infty,\inf \mathcal{I})$$ be the upper-gap and lower-gap spectrum of the capacitance operator $\mathcal{C}$ on $\ell^2$, respectively, and let $P_{\pm}$ be the corresponding spectral projections. Recall that $\mathcal{C}^{[k]}$ is the restriction of $\mathcal{C}$ to the periodic functional space $\ell_k^2$. By the discrete Bloch decomposition (\textit{cf.} \cite[Chapter 3]{conca1995fluids}), the upper-gap and lower-gap spectrum of $\mathcal{C}^{[k]}$, denoted as $\mathcal{I}_{\pm}^{[k]}$, are 
non-empty subsets of $\mathcal{I}_{\pm}$, respectively, \textit{i.e.}, $\emptyset\neq \mathcal{I}_{\pm}^{[k]}\subset \mathcal{I}_{\pm}$. We denote the spectral projections of $\mathcal{C}^{[k]}$ onto $\mathcal{I}_{\pm}^{[k]}$ by $P_{\pm}^{[k]}$, respectively. With this notation, the second lemma reads as follows, which is trivial for $p=2$, but the estimate for $p>2$ (especially for $p=4$) is required in the proof of Proposition \ref{prop_periodic_gap_solitons}. 

\begin{lemma} \label{lem_projection_lp_norm}
The spectral projection $P^{[k]}_{\pm}$ extends to $\ell^{p}_{k}\supset\ell^{2}_{k}$ for any $p\geq 2$. Its operator norm satisfies the estimate
\begin{equation} \label{eq_projection_lp_norm}
\|P^{[k]}_{\pm}\|_{\mathcal{B}(\ell^{2}_{k})}\leq N_p,
\end{equation}
where $N_p>0$ depends on $p$ and is independent of $k$.
\end{lemma}
\begin{proof}
We prove the statement only for $P^{[k]}_{+}$; the proof for $P^{[k]}_{-}$ is similar. Note that, on $\ell^{2}_{k}$, $P^{[k]}_{+}$ can be written as an integral operator
\begin{equation} \label{eq_projection_integral_expression}
(P^{[k]}_{+}a)_{\bm{n}}^{(j)}
=\sum_{\substack{1\leq j\leq d \\ \bm{m}\in\mathbb{Z}^2}}K_{\bm{n},\bm{m}}^{i,j}a_{\bm{m}}^{(j)} 
\end{equation}
with exponentially decaying kernel $K_{\bm{n},\bm{m}}\in \mathbb{M}^{d\times d}$, which can be seen by the Riesz formula \eqref{eq_riesz_formula},
\begin{equation*}
K_{\bm{n},\bm{m}}^{i,j}=\frac{i}{2\pi}\int_{\Gamma}((\mathcal{C}-\lambda)^{-1}\delta_{\bm{m}}^{(j)},\delta_{\bm{n}}^{(i)})_{\ell^2}dz,
\end{equation*}
and the exponential decay of the resolvent $(\mathcal{C}-\lambda)^{-1}$ when $\lambda\notin \text{Spec}(\mathcal{C})$. Therefore, we have
\begin{equation} \label{eq_projection_kernel_decay}
\|K_{\bm{n},\bm{m}}\|_{\mathbb{M}^{d\times d}}\leq Ce^{-\gamma|\bm{n}-\bm{m}|}, \quad \forall \bm{n},\bm{m}\in\mathbb{Z}^2,
\end{equation}
with $C,\gamma$ being independent of $k$. Moreover, the kernel $K$ is symmetric since the capacitance operator $\mathcal{C}$ is Hermitian
\begin{equation} \label{eq_projection_kernel_symmetry}
K_{\bm{n},\bm{m}}^{i,j}=K_{\bm{m},\bm{n}}^{j,i},\quad \forall 1\leq i,j\leq d,\, \bm{n},\bm{m}\in \mathbb{Z}^2 .
\end{equation}
For each $a\in \ell^p_k$, we define $P_{+}^{[k]}a$ by the integral expression \eqref{eq_projection_integral_expression}. This is well defined and can be justified as follows. Note that by the Hölder inequality,
\begin{equation*}
\begin{aligned}
\sum_{\substack{1\leq j\leq d \\ \bm{m}\in\mathbb{Z}^2}}|K_{\bm{n},\bm{m}}^{i,j}||a_{\bm{m}}^{(j)}|
&=\sum_{\substack{1\leq j\leq d \\ \bm{m}\in\mathbb{Z}^2}}|K_{\bm{n},\bm{m}}^{i,j}|^{1/q}|K_{\bm{n},\bm{m}}^{i,j}|^{1/p}|a_{\bm{m}}^{(j)}| \\
&\leq \big(\sum_{\substack{1\leq j\leq d \\ \bm{m}\in\mathbb{Z}^2}}|K_{\bm{n},\bm{m}}^{i,j}| \big)^{1/q} \big(\sum_{\substack{1\leq j\leq d \\ \bm{m}\in\mathbb{Z}^2}}|K_{\bm{n},\bm{m}}^{i,j}| |a_{\bm{m}}^{(j)}|^{p}\big)^{1/p},
\end{aligned}
\end{equation*}
where $q=\frac{p}{p-1}$. On the other hand, by the exponential decay \eqref{eq_projection_kernel_decay}, the first factor on the right-hand side is finite and uniformly bounded, \textit{i.e.},  $$C_1:=\sup_{\substack{1\leq i\leq d \\ \bm{n}\in\mathbb{Z}^2}}\sum_{\substack{1\leq j\leq d \\ \bm{m}\in\mathbb{Z}^2}}|K_{\bm{n},\bm{m}}^{i,j}| <\infty .$$ 
Hence,
\begin{equation} \label{eq_projection_lp_norm_proof_1}
\begin{aligned}
\|P^{[k]}_{+}a\|_{\ell_{k}^{p}}^{p}
&\leq \sum_{\substack{1\leq i\leq d \\ \bm{n}\in Y_k}}\big(\sum_{\substack{1\leq j\leq d \\ \bm{m}\in\mathbb{Z}^2}}|K_{\bm{n},\bm{m}}^{i,j}||a_{\bm{m}}^{(j)}| \big)^{p}
\leq C_1^{p/q} \sum_{\substack{1\leq i\leq d \\ \bm{n}\in Y_k}}\sum_{\substack{1\leq j\leq d \\ \bm{m}\in\mathbb{Z}^2}}|K_{\bm{n},\bm{m}}^{i,j}| |a_{\bm{m}}^{(j)}|^{p} \\
&=C_1^{p/q}\Big[ \sum_{\substack{1\leq i\leq d \\ \bm{n}\in Y_k}}\sum_{\substack{1\leq j\leq d \\ \bm{m}\notin Y_k}}|K_{\bm{n},\bm{m}}^{i,j}| |a_{\bm{m}}^{(j)}|^{p}+\sum_{\substack{1\leq i\leq d \\ \bm{n}\in Y_k}}\sum_{\substack{1\leq j\leq d \\ \bm{m}\in Y_k}}|K_{\bm{n},\bm{m}}^{i,j}| |a_{\bm{m}}^{(j)}|^{p} \Big] \\
&=:C_1^{p/q}( I_1+I_2 ) .
\end{aligned}
\end{equation}
We estimate $I_1$ and $I_2$ separately. First, note that by \eqref{eq_projection_kernel_decay}
\begin{equation*}
\begin{aligned}
\sum_{\substack{1\leq i\leq d \\ \bm{n}\in Y_k}}|K_{\bm{n},\bm{m}}^{i,j}|
&\leq C\sum_{\bm{n}\in Y_k} e^{-\frac{\gamma}{2}|\bm{n}-\bm{m}|}e^{-\frac{\gamma}{2}|\bm{n}-\bm{m}|}
\leq Ce^{-\frac{\gamma}{2}(|\bm{m}|-k)}\sum_{\bm{n}\in Y_k} e^{-\frac{\gamma}{2}|\bm{n}-\bm{m}|} \\
&\leq Ce^{-\frac{\gamma}{2}(|\bm{m}|-k)}\sum_{\bm{n}\in \mathbb{Z}^2} e^{-\frac{\gamma}{2}|\bm{n}-\bm{m}|}
\leq C_2 e^{-\frac{\gamma}{2}(|\bm{m}|-k)},
\end{aligned}
\end{equation*}
where $C_2>0$ is independent of $k$ and $p$. Consequently,
\begin{equation*}
\begin{aligned}
I_1&\leq C_2 \sum_{\substack{1\leq j\leq d \\ \bm{m}\notin Y_k}} |a_{\bm{m}}^{(j)}|^{p} e^{-\frac{\gamma}{2}(|\bm{m}|-k)} 
=C_2 \sum_{|\bm{l}|>0}\sum_{\substack{1\leq j\leq d \\ \bm{m}\in Y_k+k\bm{l}}} |a_{\bm{m}}^{(j)}|^{p} e^{-\frac{\gamma}{2}(|\bm{m}|-k)} \\
&\leq C_2 \sum_{|\bm{l}|>0}e^{-\frac{\gamma}{2}k|\bm{l}|}\sum_{\substack{1\leq j\leq d \\ \bm{m}\in Y_k+k\bm{l}}} |a_{\bm{m}}^{(j)}|^{p}
=C_2 \|a\|_{\ell^p_k}^p\sum_{|\bm{l}|>0}e^{-\frac{\gamma}{2}k|\bm{l}|},
\end{aligned}
\end{equation*}
where the $k-$periodicity of $a$ is applied in the last step. By an elementary calculation, it can be seen that the sum on at the right-hand side is uniformly bounded for $k\geq 1$. Hence, we know
\begin{equation} \label{eq_projection_lp_norm_proof_2}
\begin{aligned}
I_1\leq C_3\|a\|_{\ell^p_k}^p ,
\end{aligned}
\end{equation}
with $C_3>0$ being independent of $k$. The estimate of $I_2$ is simple. In fact, using \eqref{eq_projection_kernel_symmetry}, we have
\begin{equation} \label{eq_projection_lp_norm_proof_3}
\begin{aligned}
I_2= \sum_{\substack{1\leq j\leq d \\ \bm{m}\in Y_k}} |a_{\bm{m}}^{(j)}|^{p} \sum_{\substack{1\leq i\leq d \\ \bm{n}\in Y_k}}|K_{\bm{n},\bm{m}}^{i,j}|
\leq C_4\sum_{\substack{1\leq j\leq d \\ \bm{m}\in Y_k}} |a_{\bm{m}}^{(i)}|^{p} =C_4\|a\|_{\ell^p_k}^p,
\end{aligned}
\end{equation}
where $$C_4:=\sup_{\substack{1\leq j\leq d \\ \bm{m}\in\mathbb{Z}^2}}\sum_{\substack{1\leq i\leq d \\ \bm{n}\in\mathbb{Z}^2}}|K_{\bm{n},\bm{m}}^{i,j}|=\sup_{\substack{1\leq i\leq d \\ \bm{n}\in\mathbb{Z}^2}}\sum_{\substack{1\leq j\leq d \\ \bm{m}\in\mathbb{Z}^2}}|K_{\bm{n},\bm{m}}^{i,j}|=C_1 <\infty.$$ With \eqref{eq_projection_lp_norm_proof_1}-\eqref{eq_projection_lp_norm_proof_3}, we arrive at
\begin{equation} \label{eq_projection_lp_norm_proof_4}
\|P^{[k]}_{+}a\|_{\ell_{k}^{p}}\leq N_p \|a\|_{\ell_{k}^{p}},
\end{equation}
with $N_p:=C_1^{1-\frac{1}{p}}(C_1+C_3)^{\frac{1}{p}}$. Thus, $P^{[k]}_{+}$ is well defined on $\ell_k^p$. The inequality \eqref{eq_projection_lp_norm_proof_4} also implies the estimate of the norm \eqref{eq_projection_lp_norm}.
\end{proof}

The last lemma justifies, in a suitable sense, the convergence of the spectral projections $P_{\pm}^{[k]}$ to $P_{\pm}$, respectively.

\begin{lemma} \label{lem_converge_spectral_projection}
The operator $R^{[k]}P_{+}^{[k]}S^{[k]}$ converges strongly to $P_{+}$ in $\ell^2$, that is,
\begin{equation} \label{eq_converge_spectral_projection}
\|R^{[k]}P_{+}^{[k]}S^{[k]}z-P_{+}z\|_{\ell^2}\to 0 ,\quad \forall z\in \ell^2.
\end{equation}
Similarly, $P_{-}^{[k]}\overset{\ell^2-\text{strong}}{\longrightarrow}P_{-}$.
\end{lemma}

\begin{proof}
{\color{blue}Step 0.} The proof is structured as follows. We first prove that the following weak convergence of resolvents holds uniformly for $\lambda$ lying in a compact set disadjoint from the spectrum of $\mathcal{C}$:
\begin{equation} \label{eq_spectral_convergence_proof_1}
(R^{[k]}(\mathcal{C}^{[k]}-\lambda)^{-1}S^{[k]}z,b)_{\ell^2}\to ((\mathcal{C}-\lambda)^{-1}z,b)_{\ell^2},\quad \text{for any compactly supported $b$.}
\end{equation}
Then, by the Riesz formula
\begin{equation} \label{eq_riesz_formula}
P=\frac{i}{2\pi}\int_{\Gamma}(\mathcal{L}-z)^{-1}dz,\quad (P=P_{+}/P_{+}^{[k]},\mathcal{L}=\mathcal{C}/\mathcal{C}^{[k]}, \text{ resp.})
\end{equation}
($\Gamma$ is a closed complex contour separating the upper-gap spectrum $\mathcal{I}_{+}$ and lower-gap spectrum $\mathcal{I}_{-}$), we immediately have
\begin{equation*}
(R^{[k]}P^{[k]}_{+}S^{[k]}z,b)_{\ell^2}\to (P_{+}z,b)_{\ell^2}, \quad \text{for any compactly supported $b$.}
\end{equation*}
Based on this result and following the lines of the proof to \eqref{eq_local_bound_weak_converge_2} in Lemma \ref{lem_local_bound_weak_converge}, one obtains the following $\ell^2-$weak convergence:
\begin{equation} \label{eq_spectral_convergence_proof_2}
(R^{[k]}P^{[k]}_{+}S^{[k]}z,b)_{\ell^2}\to (P_{+}z,b)_{\ell^2}, \quad \text{for any $b\in \ell^2$.}
\end{equation}
We then prove \eqref{eq_converge_spectral_projection} using \eqref{eq_spectral_convergence_proof_2}.

{\color{blue}Step 1.} Denote $x^{[k]}:=(\mathcal{C}^{[k]}-\lambda)^{-1}S^{[k]}z$. To prove \eqref{eq_spectral_convergence_proof_1}, it is sufficient to show that any subsequence of $x^{[k]}$ admits a further subsequence that weakly converges to $(\mathcal{C}-\lambda)^{-1}z \in \ell^2(\mathbb{Z}^2;\mathbb{C}^{d})$ in the sense of \eqref{eq_spectral_convergence_proof_1}. Let $x^{[k_i]}$ be a subsequence of $x^{[k]}$. Note that $\|x^{[k_i]}\|_{\ell_{k_i}^2}$ is uniformly bounded since $\Gamma$ is disadjoint from the upper-gap spectrum:
\begin{equation} \label{eq_spectral_convergence_proof_3}
\|x^{[k_i]}\|_{\ell_{k_i}^2}\leq \|(\mathcal{C}^{[k_i]}-\lambda)^{-1}\|\|S^{[k_i]}z\|_{\ell^2_{k_i}}\leq \frac{1}{\text{dist}(\Gamma,\mathcal{I}_{+})}\|z\|_{\ell^2}.
\end{equation}
Thus, by Lemma \ref{lem_local_bound_weak_converge}, $x^{[k_i]}$ admits an $\ell^{\infty}$ weakly convergent subsequence, still denoted as $x^{[k_i]}$ for ease of notation, with the limit $x\in \ell^2$. Now, we prove $x=(\mathcal{C}-\lambda)^{-1}z$. Note that $x^{[k_i]}$ solves the following problem:
\begin{equation*}
    \mathcal{C}x^{[k_i]}-\lambda x^{[k_i]}=S^{[k_i]}z \quad \text{in $\mathbb{Z}^2$}.
\end{equation*}
Taking the inner product with a compactly supported test function $b$ yields
\begin{equation} \label{eq_spectral_convergence_proof_4}
(\mathcal{C}x^{[k_i]},b)_{\ell^2}-\lambda (x^{[k_i]},b)_{\ell^2}=(S^{[k_i]}z,b)_{\ell^2} .
\end{equation}
By the exponential decay of the capacitance operator shown in Proposition \ref{prop_cap_operator}, the fact that $\mathcal{C}$ is symmetric, and the Fubini theorem, it follows that the product $(\mathcal{C}x^{[k_i]},b)_{\ell^2}$ is equal to $(x^{[k_i]},\mathcal{C}b)_{\ell^2}$ and we have
\begin{equation*}
(\mathcal{C}x^{[k_i]},b)_{\ell^2}=(x^{[k_i]},\mathcal{C}b)_{\ell^2}
\to (x,\mathcal{C}b)_{\ell^2}
=(\mathcal{C}x,b)_{\ell^2},
\end{equation*}
where we have applied the weak convergence of $x^{[k_i]}$ and the fact that $\mathcal{C}b\in \ell^1$ (this, again, can be checked using the exponential decay of $\mathcal{C}$ and the fact that $b$ is compactly supported). Similarly, we have
\begin{equation*}
(x^{[k_i]},b)_{\ell^2}\to (x,b)_{\ell^2} ,
\end{equation*}
and
\begin{equation*}
(S^{[k_i]}z,b)_{\ell^2}\to (z,b)_{\ell^2} ,
\end{equation*}
where the last equality follows from the definition of the periodization operator $S^{[k_i]}$. As a consequence, by passing $i\to\infty$ in equality \eqref{eq_spectral_convergence_proof_4}, we obtain
\begin{equation*}
(\mathcal{C}x,b)_{\ell^2}-\lambda (x,b)_{\ell^2}=(z,b)_{\ell^2}.
\end{equation*}
This implies that $x=(\mathcal{C}-\lambda)^{-1}z$ since the solution to the above equation is unique (recall that $\lambda\notin \text{Spec}(\mathcal{C})$), and concludes the proof of \eqref{eq_spectral_convergence_proof_1}.

{\color{blue}Step 2.} Finally, we prove \eqref{eq_converge_spectral_projection} using \eqref{eq_spectral_convergence_proof_2}. The square of the left side of \eqref{eq_converge_spectral_projection} equals
\begin{equation} \label{eq_spectral_convergence_proof_5}
\begin{aligned}
&\|R^{[k]}P^{[k]}_{+}S^{[k]}z-P_{+}z\|_{\ell^2}^2 \\
&=(R^{[k]}P^{[k]}_{+}S^{[k]}z,R^{[k]}P^{[k]}_{+}S^{[k]}z)_{\ell^2}-2(R^{[k]}P^{[k]}_{+}S^{[k]}z,P_{+}z)_{\ell^2}+\|P_{+}z\|_{\ell^2}^2 \\
&=(P^{[k]}_{+}S^{[k]}z,P^{[k]}_{+}S^{[k]}z)_{\ell^2_{k}}-2(R^{[k]}P^{[k]}_{+}S^{[k]}z,P_{+}z)_{\ell^2}+\|P_{+}z\|_{\ell^2}^2 \\
&=(P^{[k]}_{+}S^{[k]}z,S^{[k]}z)_{\ell^2_{k}}-2(R^{[k]}P^{[k]}_{+}S^{[k]}z,P_{+}z)_{\ell^2}+\|P_{+}z\|_{\ell^2}^2 \\
&=(R^{[k]}P^{[k]}_{+}S^{[k]}z,P_{+}z)_{\ell^2}-2(R^{[k]}P^{[k]}_{+}S^{[k]}z,P_{+}z)_{\ell^2}+\|P_{+}z\|_{\ell^2}^2 \\
&=-(R^{[k]}P^{[k]}_{+}S^{[k]}z,P_{+}z)_{\ell^2}+\|P_{+}z\|_{\ell^2}^2.
\end{aligned}
\end{equation}
By the weak convergence \eqref{eq_spectral_convergence_proof_2}, 
we directly obtain the following convergence since $P_{+}z\in \ell^2$, which then,  with \eqref{eq_spectral_convergence_proof_5}, concludes the proof of \eqref{eq_converge_spectral_projection}
\begin{equation*}
(R^{[k]}P^{[k]}_{+}S^{[k]}z,P_{+}z)_{\ell^2}\to \|P_{+}z\|_{\ell^2}^2 .
\end{equation*}
\end{proof}

\subsubsection{Detailed proof}

Now, we embark on proving Proposition \ref{prop_periodic_gap_solitons}. Inspired by the minimax argument for studying periodic gap solitons \cite{pankov2010soliton_discrete,pelinovsky2011localization}, the key step is to construct an admissible set on which the landscape of the functional $J^{[k]}$ possesses a valley-top structure. The following result holds. 
\begin{lemma} \label{lem_admissible_set}
There exists $z_0\in \ell^2$, $k_0\in \mathbb{N},$ and $M_0>0$ such that when $k> k_0$ and $\|V\|_{\mathcal{B}(\ell^{\infty},\ell^{1})}<M_0$, there is a non-empty bounded open set $M^{[k]}\subset \ell_k^2$ that satisfies
\begin{itemize}
    \item[(i)] $\|M^{[k]}\|_{\ell^2_k}:=\sup_{a\in M^{[k]}}\|a\|_{\ell^2_k}\leq M_1<\infty$;
    \item[(ii)] $\inf_{a\in M^{[k]}}(R^{[k]}a,z_0)_{\ell^2}\geq m_1>0$;
    \item[(iii)] $\sup_{\partial M^{[k]}} J^{[k]}<\sup_{M^{[k]}} J^{[k]}<\infty$,
\end{itemize}
where $m_1,M_1>0$ are independent of $k$.
\end{lemma}

Since $J^{[k]}$ is a real-valued continuous functional defined on the finite-dimensional space $\ell_k^2\simeq \mathbb{C}^{k^2d}$, it attains its maximum on the compact set $M^{[k]}$. Moreover, property (iii) in Lemma \ref{lem_admissible_set} guarantees that the maximum point lies in the interior of $M^{[k]}$, which is hence a critical point of $J^{[k]}$ and is denoted by $a^{[k]}$. Then the properties of the critical point $a^{[k]}$ claimed in Proposition \ref{prop_periodic_gap_solitons} follow from that of the set $M^{[k]}$.

Note that one can directly apply the minimax theorems (\textit{cf.} \cite{Pankov2005soliton_continuous,benci1979critical}), using properties (i) and (iii) in Lemma \ref{lem_admissible_set}, to conclude the existence of critical points of $J^{[k]}$ in $\ell_k^2$. However, the original statements of the minimax theorems do not characterize the position of critical points, \textit{e.g.}, the property (ii) in Lemma \ref{lem_admissible_set}, due to the infinite-dimensional setups. Due to the discrete structure of the problem, our functional $J^{[k]}$ is defined on a finite-dimensional space. This allows us to locate the critical points $J^{[k]}$ (exactly in $M^{[k]}$ as shown above) and obtain (ii). As indicated in Section \ref{subsec_soliton_road}, this condition is critical to guarantee the nontriviality of solution we will find later.

\begin{proof}[Proof of Lemma \ref{lem_admissible_set}]
{\color{blue}Step 1.} We first construct the set $M^{[k]}$. To start, we fix a unit vector $z_0\in \ell^2$ in the range of $P_{+}$:
\begin{equation} \label{eq_periodic_gap_solitons_proof_8}
z_0\in \text{Ran}(P_{+}),\quad \|z_0\|_{\ell^2}=1.
\end{equation}
Then we define
\begin{equation*}
z_0^{[k]}:=\frac{1}{\|P_{+}^{[k]}S^{[k]}z_0\|_{\ell^2_k}}P_{+}^{[k]}S^{[k]}z_0 .
\end{equation*}
With this notation, the set $M^{[k]}$ is defined as
\begin{equation*}
M^{[k]}:=\{a\in \text{Ran}(P_{-}^{[k]})\oplus \mathbb{R}^{+}z_0^{[k]}:\,\|a\|_{\ell^2_k}<\rho,\, (a,z_0^{[k]})_{\ell^2_k}>r\}
\end{equation*}
with $\rho,r>0$ to be determined. In the following steps, we verify that $M^{[k]}$ satisfies properties (i)-(iii) in Lemma \ref{lem_admissible_set} when $\rho,r>0$ are appropriately chosen. In fact, properties (i) and (ii) are directly satisfied for any $k-$independent parameters $\rho,r>0$ by setting $M_1:=\rho$, $m_1=\frac{1}{2}r$ and choosing $k_{0}$ to be sufficiently large such that $\|P^{[k]}_{+}S^{[k]}z_{0}\|_{\ell^2_k}>\frac{2}{3}$ (this is achievable because $\|P^{[k]}_{+}S^{[k]}z_{0}\|_{\ell^2_k}=\|R^{[k]}P^{[k]}_{+}S^{[k]}z_{0}\|_{\ell^2}\to \|z_{0}\|_{\ell^2}=1$ by Lemma \ref{lem_converge_spectral_projection}). The key point is to appropriately select $\rho$ and $r$ such that property (iii) holds. Denote the isolation distance between $\lambda$ and the spectrum of the capacitance operator as $\delta:=\text{dist}\{\lambda,\text{Spec}(\mathcal{C})\}>0$. We claim that if the norm of the defect $V$ is controlled by
\begin{equation*}
\|V\|_{\mathcal{B}(\ell^{\infty},\ell^{1})}<M_0:=\frac{\delta}{2},
\end{equation*}
then there exists a sufficiently large $\rho>0$ and a small $r<0$ such that (ii) holds. This is proved using the results from the following steps.

{\color{blue}Step 2.} First we show that, for any $r<\sqrt{\frac{\delta}{2\lambda\sigma}}$, the vector $tz_{0}^{[k]}\in M^{[k]}$ with $t=\sqrt{\frac{\delta}{2\lambda\sigma}}$ satisfies
\begin{equation} \label{eq_periodic_gap_solitons_proof_9}
J^{[k]}(tz_{0}^{[k]})=\frac{\delta^2}{16\lambda\sigma}>0 .
\end{equation}
In fact, the positive definitiveness $\mathcal{C}^{[k]}-\lambda\gg \delta$ on $\text{Ran}(P_{+}^{[k]})$ implies that
\begin{equation*}
\begin{aligned}
J^{[k]}(tz_{0}^{[k]})&=\frac{t^2}{2}(\mathcal{C}^{[k]}z_{0}^{[k]},z_{0}^{[k]})_{\ell^2_{k}}-\frac{t^2}{2}\lambda\|z_{0}^{[k]}\|_{\ell^2_{k}}^2+\frac{t^2}{2}(S^{[k]}VR^{[k]}z_{0}^{[k]},z_{0}^{[k]})_{\ell^2_k}-\frac{t^4}{4}\lambda\sigma \|z_{0}^{[k]}\|_{\ell^4_{k}}^4 \\
&\geq \frac{t^2\delta}{2}\|z_{0}^{[k]}\|_{\ell^2_{k}}^2+\frac{t^2}{2}(VR^{[k]}z_{0}^{[k]},R^{[k]}z_{0}^{[k]})_{\ell^2}-\frac{t^4}{4}\lambda\sigma \|z_{0}^{[k]}\|_{\ell^2_{k}}^4 \\
&\geq \frac{t^2\delta}{2}\|z_{0}^{[k]}\|_{\ell^2_{k}}^2-\frac{t^2\|V\|_{\mathcal{B}(\ell^\infty,\ell^1)}}{2}\|R^{[k]}z_{0}^{[k]}\|_{\ell^{\infty}}^2-\frac{t^4}{4}\lambda\sigma \|z_{0}^{[k]}\|_{\ell^2_{k}}^4 \\
&\geq \frac{t^2\delta}{2}\|z_{0}^{[k]}\|_{\ell^2_{k}}^2-\frac{t^2\|V\|_{\mathcal{B}(\ell^\infty,\ell^1)}}{2}\|z_{0}^{[k]}\|_{\ell^2_{k}}^2-\frac{t^4}{4}\lambda\sigma \|z_{0}^{[k]}\|_{\ell^2_{k}}^4 \\
&\geq \frac{\delta}{4}t^2\|z_{0}^{[k]}\|_{\ell^2_{k}}^2-\frac{1}{4}\lambda\sigma t^4\|z_{0}^{[k]}\|_{\ell^4_{k}}^4,
\end{aligned}
\end{equation*}
where the elementary inequality $\|a\|_{\ell^\infty}\leq\|a\|_{\ell^4_{k}}\leq \|a\|_{\ell^2_{k}}$ and the normalization $\|z_{0}^{[k]}\|_{\ell^2_k}=1$ are applied. This concludes the proof of \eqref{eq_periodic_gap_solitons_proof_9}.

{\color{blue}Step 3.} Next, we estimate the value of $J^{[k]}$ on $\partial M^{[k]}$, from which we conclude the proof of property (iii) together with \eqref{eq_periodic_gap_solitons_proof_9}. First note that
\begin{equation*}
\begin{aligned}
J^{[k]}(y+t z_0^{[k]})
&=\frac{1}{2}(\mathcal{C}^{[k]}(y+t z_0^{[k]}),y+t z_0^{[k]})_{\ell^2_{k}}-\frac{1}{2}\lambda\|y+t z_0^{[k]}\|_{\ell^2_{k}}^2+\frac{1}{2}(S^{[k]}VR^{[k]}(t+z_{0}^{[k]}),t+z_{0}^{[k]})_{\ell^2_k} \\
&\quad -\frac{1}{4}\lambda\sigma \|y+t z_0^{[k]}\|_{\ell^4_{k}}^4 \\
&\leq \frac{1}{2}(\mathcal{C}^{[k]}y,y)_{\ell^2_{k}}-\frac{1}{2}\lambda\|y\|_{\ell^2_{k}}^2 +\frac{t^2}{2}(\mathcal{C}^{[k]}z_0^{[k]},z_0^{[k]})_{\ell^2_{k}}-\frac{t^2}{2}\lambda\|z_0^{[k]}\|_{\ell^2_{k}}^2 \\
&\quad +\frac{\delta}{4}\|y+tz_{0}^{[k]}\|_{\ell^2_{k}}^2 -\frac{1}{4}\lambda\sigma \|y+t z_0^{[k]}\|_{\ell^4_{k}}^4 \\
&= \frac{1}{2}(\mathcal{C}^{[k]}y,y)_{\ell^2_{k}}-\frac{1}{2}\lambda\|y\|_{\ell^2_{k}}^2 +\frac{t^2}{2}(\mathcal{C}^{[k]}z_0^{[k]},z_0^{[k]})_{\ell^2_{k}}-\frac{t^2}{2}\lambda\|z_0^{[k]}\|_{\ell^2_{k}}^2 \\
&\quad +\frac{\delta}{4}(\|y\|_{\ell^2_k}^2+t^2\|z_0^{[k]}\|_{\ell^2_k}^2) -\frac{1}{4}\lambda\sigma \|y+t z_0^{[k]}\|_{\ell^4_{k}}^4,
\end{aligned}
\end{equation*}
where the orthogonality $y\perp z_0^{[k]}$ and the invariance of $\mathcal{C}^{[k]}$ on $\text{Ran}(P_{\pm}^{[k]})$ are applied. By the negative definitiveness of $\mathcal{C}^{[k]}-\lambda\ll -\delta$ on $\text{Ran}(P_{-}^{[k]})$ and the boundedness $$\|\mathcal{C}^{[k]}\|_{\mathcal{B}(\ell_k^2)}\leq \|\mathcal{C}\|_{\mathcal{B}(\ell^2)}<\infty$$ (since $\text{Spec}(\mathcal{C}^{[k]})\subset \text{Spec}(\mathcal{C})$), we see that
\begin{equation*}
\begin{aligned}
J^{[k]}(y+t z_0^{[k]})
&\leq -\frac{\delta}{4}\|y\|_{\ell^2_k}^2+C_1 t^2 -\frac{1}{4}\lambda\sigma \|y+t z_0^{[k]}\|_{\ell^4_{k}}^4,
\end{aligned}
\end{equation*}
where $C_1:=\frac{1}{2}(\|\mathcal{C}\|_{\mathcal{B}(\ell^2)}+|\lambda|+\frac{\delta}{2})$. Note that $tz_0^{[k]}=P^{[k]}_{+}(y+t z_0^{[k]})$. Hence, by Lemma \ref{lem_projection_lp_norm},
\begin{equation} \label{eq_periodic_gap_solitons_proof_10}
\begin{aligned}
J^{[k]}(y+t z_0^{[k]})
\leq -\frac{\delta}{4}\|y\|_{\ell^2_k}^2+C_1 t^2 -C_2 t^4 \| z_0^{[k]}\|_{\ell^4_{k}}^4,
\end{aligned}
\end{equation}
where $C_2:=\frac{\lambda\sigma}{4N_{4}^{-4}}$ ($N_4$ is introduced in \eqref{eq_projection_lp_norm}). An important observation is that Lemma \ref{lem_converge_spectral_projection} implies that
\begin{equation*}
\lim_{k\to\infty}\|z_0^{[k]}\|_{\ell_k^4}
=\lim_{k\to\infty}\|R^{[k]}P^{[k]}_{+}S^{[k]}z_0\|_{\ell^4}
=\|z_0\|_{\ell^4} > 0.
\end{equation*}
Hence, there exists $k_0\in\mathbb{N}$ such that for $k>k_0$, it holds that $\|z_0^{[k]}\|_{\ell_k^4}>\frac{1}{2}\|z_0\|_{\ell^4}$. This implies the coefficient of $-t^4$ in \eqref{eq_periodic_gap_solitons_proof_10} is nonzero (this helps to control the upper bound $\sup_{\partial M^{[k]}}J^{[k]}$ as will be seen later):
\begin{equation} \label{eq_periodic_gap_solitons_proof_11}
\begin{aligned}
J^{[k]}(y+t z_0^{[k]})
\leq -\frac{\delta}{4}\|y\|_{\ell^2_k}^2+C_1 t^2 -C_3 t^4 , 
\end{aligned}
\end{equation}
where $C_3:=\frac{1}{2}C_2\|z_0\|_{\ell^4}^4>0$.

Now, with \eqref{eq_periodic_gap_solitons_proof_11}, we can estimate the boundary value of $J^{[k]}$. Note that
\begin{equation*}
\begin{aligned}
\partial M^{[k]}
&=\big\{y+t z_0^{[k]}\in \text{Ran}(P_{-}^{[k]})\oplus \mathbb{R}^{+}z_0^{[k]}:\,t=\sqrt{\rho^2-\|y\|_{\ell^2_k}^2}> 0\big\} \\
&\quad  \bigcup \big\{y+r z_0^{[k]}\in \text{Ran}(P_{-}^{[k]})\oplus \mathbb{R}^{+}z_0^{[k]}:\, \|y\|_{\ell^2_k}^2\leq\rho^2-r^2\big\}.
\end{aligned}
\end{equation*}
When $t=\sqrt{\rho^2-\|y\|_{\ell^2_k}^2}> 0$, we see from \eqref{eq_periodic_gap_solitons_proof_11} that
\begin{equation*}
J^{[k]}(y+t z_0^{[k]})
\leq -\frac{\delta}{4}\rho^2+(C_1+\frac{\delta}{4}) t^2 -C_3 t^4 .
\end{equation*}
Since $(C_1+\frac{\delta}{4}) t^2 -C_3 t^4$ is bounded from above as $t\to\infty$ (thanks to the fact that $C_3>0$), we can take $\rho$ to be sufficiently large such that $J^{[k]}(y+t z_0^{[k]}) \leq 0$. On the other hand, we further impose that
\begin{equation*}
r<\sqrt{\frac{\delta^2}{32C_1\lambda\sigma}}.
\end{equation*}
As a consequence, when $t=r$ and $\|y\|_{\ell^2_k}^2\leq\rho^2-r^2$, \eqref{eq_periodic_gap_solitons_proof_11} implies that
\begin{equation*}
J^{[k]}(y+r z_0^{[k]})
\leq C_1 t^2 <\frac{\delta^2}{32\lambda\sigma}.
\end{equation*}
In conclusion, by selecting $$\rho>\sup_{t\in\mathbb{R}}\Big\{(C_1+\frac{\delta}{4}) t^2 -C_3 t^4 \Big\}\quad \text{and} \quad r<\sqrt{\frac{\delta^2}{32C_1\lambda\sigma}},$$ it holds that
\begin{equation*}
\sup_{\partial M^{[k]}} J^{[k]} \leq \frac{\delta^2}{32\lambda\sigma} .
\end{equation*}
This, together with \eqref{eq_periodic_gap_solitons_proof_9}, gives
\begin{equation*}
\sup_{\partial M^{[k]}} J^{[k]} \leq\frac{\delta^2}{32\lambda\sigma}<\frac{\delta^2}{16\lambda\sigma}\leq \sup_{M^{[k]}} J^{[k]},
\end{equation*}
and concludes the proof of property (iii).
\end{proof}
 
\subsection{Proof of Proposition \ref{prop_concentration}} \label{subsec_soliton_concentration}

In this section, we prove Proposition \ref{prop_concentration} using a standard compactness argument. First, for the sequence $a^{[k]}$ introduced in Proposition \ref{prop_periodic_gap_solitons}, its boundedness \eqref{eq_periodic_gap_solitons_1} and the weak compactness shown in Lemma \ref{lem_local_bound_weak_converge} indicate that we can extract a subsequence $a^{[k_i]}$ such that
\begin{equation}
\label{eq_concentration_proof_1}
a^{[k_i]}\overset{\ell^{\infty}-\text{weak}}{\rightharpoonup}a,\quad R^{[k_i]}a^{[k_i]}\overset{\ell^{2}-\text{weak}}{\rightharpoonup}a
\quad \text{with $\|a\|_{\ell^2}\leq M_1$.}
\end{equation}
We first show that $a\neq 0$. Thanks to the lower bound estimate in \eqref{eq_periodic_gap_solitons_1} and the weak convergence of $R^{[k_i]}a^{[k_i]}$, we see that
\begin{equation*}
m_1\leq (R^{[k_i]}a^{[k_i]},z_0)_{\ell^2}\to (a,z_0)_{\ell^2},
\end{equation*}
which implies that $a$ is nontrivial.

Next, we verify that $a$ solves equation \eqref{eq_discrete_station_eq_defect}. By Proposition \ref{prop_periodic_gap_solitons}, each $a^{[k_i]}$ satisfies
\begin{equation} \label{eq_concentration_proof_2}
\mathcal{C}^{[k_i]}a^{[k_i]}+S^{[k_i]}VR^{[k_i]}a^{[k_i]}-\lambda a^{[k_i]}-\lambda\sigma |a^{[k_i]}|^2 a^{[k_i]}=0 .
\end{equation}
We claim that for any compactly supported test function $b$,
\begin{equation} \label{eq_concentration_proof_3}
\begin{aligned}
&(\mathcal{C}^{[k_i]}a^{[k_i]},b)_{\ell^2}\to (\mathcal{C}a,b)_{\ell^2},\quad
(S^{[k_i]}VR^{[k_i]}a^{[k_i]},b)_{\ell^2}\to (Va,b)_{\ell^2},\\
&(a^{[k_i]},b)_{\ell^2}\to (a,b)_{\ell^2},\quad
(|a^{[k_i]}|^2 a^{[k_i]},b)_{\ell^2}\to (|a|^2 a,b)_{\ell^2}.
\end{aligned}
\end{equation}
Then, multiplying \eqref{eq_concentration_proof_2} with Kronecker function $b=\delta_{\bm{n}}^{(i)}$ and passing $i\to\infty$, we see that $a$ solves  pointwise the following nonlinear equation:
\begin{equation} \label{eq_concentration_proof_4}
\mathcal{C}a+Va-\lambda(1+\sigma|a|^2)a=0.
\end{equation}
Since $b$ is compactly supported, both convergences $$(a^{[k_i]},b)_{\ell^2}\to (a,b)_{\ell^2} \quad \text{and} \quad  (|a^{[k_i]}|^2 a^{[k_i]},b)_{\ell^2}\to (|a|^2 a,b)_{\ell^2}$$ follow directly from the weak convergence \eqref{eq_concentration_proof_1}. On the other hand, following the same lines as in Step 1 of Lemma \ref{lem_converge_spectral_projection}, we can prove that $$(\mathcal{C}^{[k_i]}a^{[k_i]},b)_{\ell^2}\to (\mathcal{C}a,b)_{\ell^2}.$$ Hence, we only need to prove that $(S^{[k_i]}VR^{[k_i]}a^{[k_i]},b)_{\ell^2}\to (Va,b)_{\ell^2}$. Note that by the definition of the periodization and truncation operators, and the self-adjointness of $V$, we have
\begin{equation*}
(S^{[k_i]}VR^{[k_i]}a^{[k_i]},b)_{\ell^2}=(VR^{[k_i]}a^{[k_i]},R^{[k_i]}b)_{\ell^2}=(R^{[k_i]}a^{[k_i]},VR^{[k_i]}b)_{\ell^2}.
\end{equation*}
Applying the weak convergence $a^{[k_i]}\overset{\ell^{\infty}-\text{weak}}{\rightharpoonup}a$ and the fact that $R^{[k_i]}b=b$ for sufficiently large $i$, we see
\begin{equation*}
(S^{[k_i]}VR^{[k_i]}a^{[k_i]},b)_{\ell^2}=(R^{[k_i]}a^{[k_i]},VR^{[k_i]}b)_{\ell^2}\to (a,Vb)_{\ell^2}=(Va,b)_{\ell^2}.
\end{equation*}
This concludes the proof of \eqref{eq_concentration_proof_3}.


Finally, since $a\in \ell^2(\mathbb{Z}^2;\mathbb{C}^{d})$, \eqref{eq_concentration_proof_4} implies that $a$ can be seen as the eigenfunction of the operator $\tilde{\mathcal{C}}:=\mathcal{C}-\lambda\sigma |a|^2+V$. Note that $\tilde{\mathcal{C}}$ is an $\ell^2$-bounded perturbation of $\mathcal{C}$. Therefore, $a$ is an in-gap defect mode corresponding to the unperturbed operator $\mathcal{C}$. In that case, it is well known that $a$ decays exponentially at infinity \cite{pankov2010soliton_discrete}.

\section{Concluding remarks} \label{sec:conclusion}

In this paper, we have developed a comprehensive mathematical framework toward analyzing nonlinear resonances in periodic systems of high-contrast subwavelength resonators. We have derived a characterization of soliton-like solutions and established their existence in both full- and half-space crystals. In order to do so, we have established a tight-binding approximation of the linear capacitance operator in subwavelength resonator crystals. Our results in this paper would allow us to study the topological properties of periodic lattices of subwavelength nonlinear resonators, such as the emergence of nonlinearity-induced topological edge states, and to also elucidate the interplay between nonlinearity and disorder. This would be the subject of forthcoming works. 

\subsection*{Acknowledgments} This work was supported in part by the Swiss National Science Foundation grant number 2025-10004276. It was initiated while H.A. was visiting the Hong Kong Institute for Advanced Study as a Senior Fellow. 

\footnotesize
\bibliographystyle{plain}
\bibliography{ref}

\end{document}